\documentclass[preprint]{elsarticle}
\usepackage{lineno,hyperref}
\usepackage{amsmath,xcolor,float}
\usepackage[top=0.5in, left=0.7in, bottom=0.5in, right=0.7in]{geometry}
\usepackage{amsfonts}
\usepackage{epstopdf} 
\usepackage{epsfig}
\usepackage{natbib}
\usepackage{amssymb}
\modulolinenumbers[5]

\newtheorem{theorem}{Theorem}

\newtheorem{property}{Property}

\newtheorem{corollary}{Corollary}

\newtheorem{proposition}{Proposition}

\newenvironment{proof}[1][Proof]{\noindent\textbf{#1.} }{\ \rule{0.5em}{0.5em}}

\biboptions{authoryear}

\makeatletter
\def\ps@pprintTitle{%
  \let\@oddhead\@empty
  \let\@evenhead\@empty
  \let\@oddfoot\@empty
  \let\@evenfoot\@oddfoot
}
\makeatother

\begin{document}

\begin{frontmatter}

\title{Ambiguity aversion as a route to randomness in a duopoly game}

\author{D. Radi$^{a}$ and L. Gardini$^{b}$}

\cortext[mycorrespondingauthor]{Corresponding author: Davide Radi, Via Necchi 9, 20123, Milan, Italy.  Email: davide.radi@unicatt.it  }

\address[mymainaddress]{DiMSEFA, Catholic University of the Sacred Heart, Milan, Italy, and Department of Finance,\\
 V\v{S}B - Technical University of Ostrava, Ostrava, Czech Republic.}
\address[mysecondaryaddress]{Department of Economics, Society and Politics, University of Urbino Carlo Bo, Urbino, Italy.}

\begin{abstract}
The global dynamics is investigated for a duopoly game where the perfect foresight hypothesis is relaxed and firms are worst-case maximizers. Overlooking the degree of product substitutability as well as the sensitivity of price to quantity, the unique and globally stable Cournot-Nash equilibrium of the complete-information duopoly game, loses stability when firms are not aware if they are playing a duopoly game, as it is, or an oligopoly game with more than two competitors. This finding resembles Theocharis' condition for the stability of the Cournot-Nash equilibrium in oligopolies without uncertainty. As opposed to complete-information oligopoly games, coexisting attractors, disconnected basins of attractions and chaotic dynamics emerge when the Cournot-Nash equilibrium loses stability. This difference in the global dynamics is due to the nonlinearities introduced by the worst-case approach to uncertainty, which mirror in bimodal best-reply functions. Conducted with techniques that require a symmetric setting of the game, the investigation of the dynamics reveals that a chaotic regime prevents firms from being ambiguity averse, that is, firms are worst-case maximizers only in the quantity-expectation space. Therefore, chaotic dynamics are the result and at the same time the source of profit uncertainty.
\end{abstract}



\end{frontmatter}


\section{Introduction}

Ambiguity, or uncertainty, is the inability of the decision-maker to formulate a unique probability or his lack of trust in any single probability estimate. As demonstrated in his seminal contribution by \cite{Ellsberg1961}, ambiguity affects the choice of economic agents. Taking a step further in this direction, the current paper analyzes the impact of ambiguity aversion on the quantity dynamics of a duopoly game when the perfect foresight hypothesis is relaxed.

Agents' ambiguity aversion is accommodated through the maxmin expected utility model of \cite{GilboaSchmeidler1989}. According to this model, an agent maximizes the minimum expected value taken over a set of probability distribution functions that represent his beliefs about the state of the world. In case ambiguity regards all the possible probability distributions over a bounded set of possible realizations, we obtain a distribution-free decision model known as the worst-case approach to uncertainty, see \cite{Ben-TalGhaouiNemirovski2009}. A game where players adopt a worst-case approach to uncertainty is denoted as game with ambiguity aversion and an equilibrium of this game is referred to as a robust-optimization equilibrium, see \cite{AghassiBertsimas2006} and \cite{CrespiRadiRocca2017}. Also referred to as a robust game, a game with ambiguity aversion can be seen as an extreme version of an incomplete information game with multiple priors, see, e.g., \cite{KajiiUi2005}, and it must not be confused with an ambiguous game, see, e.g., \cite{Marinacci2000}, where the Choquet expected utility model is employed.

Characterized by constant marginal costs and linear and downward-sloping inverse demand functions, a (symmetric) duopoly game with ambiguity aversion is considered, which generalizes the complete information version in \cite{SinghVives1984}. The sensitivity of price to quantity is uncertain as well as the degree of substitutability of the products and the composition of the industry. Consistent with their belief about the set of uncertainty, firms are worst-case maximizers as in \cite{CrespiRadiRocca2020}. As opposed to \cite{CrespiRadiRocca2020}, the perfect foresight assumption is relaxed and firms predict the next-period level of production of the competitors according to a constant expectation scheme. Then, starting from a given initial condition, the producers simultaneously update their productions at each discrete time period, which is the so-called \emph{Cournot t\^{a}tonnement} of the duopoly game.

Assuming a bounded set of possible values of the degree of substitutability of the products and of the sensitivity of price to quantity, the Cournot-Nash equilibrium loses stability when firms ignore whether they are playing a duopoly game or an oligopoly game with more than two competitors. This finding generalizes Theocharis' result, according to which the stability of the Cournot-Nash equilibrium is precluded in a complete-information oligopoly with more than three firms, see \cite{Theocharis1960}. Differently from an oligopoly with complete information and more than three firms, in a duopoly game with ambiguity aversion the loss of stability of the Cournot-Nash equilibrium is associated with chaotic dynamics, such as coexisting cyclical chaotic attractors, and disconnected basins of attraction. This is due to the nonlinearities introduced by the worst-case approach to uncertainty, which reflects in best-reply functions that become bimodal, therefore non-monotonically decreasing.

Due to the constant expectation scheme, the quantity dynamics of the duopoly game with ambiguity aversion follows a decoupled square two-dimensional system, see \cite{BischiMammanaGardini2000}, and, due to the symmetric setting, the subspace represented by firms producing the same level of output is an invariant set. The dynamics along this invariant set is related to that of the best-reply function, that is, to a one-dimensional map which is piecewise linear and non-invertible. By means of the skew-tent map as a normal form, see \cite{SushkoAvrutinGardini2015}, we describe the dynamics on the equal-level-of-production restriction of our duopoly game with ambiguity aversion. These results and some properties of decoupled square two-dimensional systems are then employed to characterize the sequence of cyclical chaotic attractors that emerge outside the equal-level-of-production restriction when the Cournot-Nash equilibrium loses stability.

Chaotic dynamics in duopoly models are not a novelty. Inspired by \cite{Rand1978}, there are several Cournot duopoly models that exhibit complicated dynamics. A first group of models is characterized by non-monotonic best-reply functions related to sophisticated cost functions, see, e.g., \cite{Kopel1996}, \cite{BischiMammanaGardini2000} and \cite{BischiKopel2001}. A second group of models is characterized by non-monotonic best-reply functions originated by sophisticated demand functions, such as the isoelastic demand curve, see, e.g., \cite{Puu1991} and \cite{TramontanaGardiniPuu2010}. Finally, a third group of models is characterized by either sophisticated dynamic adjustment mechanisms, see, e.g., \cite{BischiNaimzada2000}, or evolutionary selection processes to discriminate between expectation schemes, see, e.g., \cite{DrosteHommesTuinstra2002}. The current work adds to this literature by showing that complicated dynamics can emerge even in a very stylized duopoly model when uncertainty is combined with a constant expectation scheme and firms are worst-case maximizers. In addition to that, a peculiarity of the proposed duopoly game with ambiguity aversion is an abrupt transition from a globally stable Cournot-Nash equilibrium to a chaotic dynamics, which is obtained by perturbing the configuration of the uncertainty set. To the best of our knowledge, this is the first attempt to study the nonlinear dynamics that emerge in a duopoly game with ambiguity aversion when the perfect foresight hypothesis is relaxed.

A further aspect is worth noting. The worst-case approach to uncertainty ensures a maximum-guaranteed payoff conditionally on the realization of the expected competitor's next-period level of production. In a chaotic regime the accuracy of the expectation is reduced by the high volatility of the quantity dynamics and the maximum-guaranteed payoff becomes uncertain on its own. Indeed, under chaotic dynamics, firms are worst-case maximizers in the quantity-expectation space but are not so in the action space. Specifically, we observe a realized worst-case profit which is not the maximum-guaranteed payoff. Moreover, it is neither the maximum-guaranteed expected payoff, nor the best-possible expected payoff, nor is it a value in between. This is an extra source of profit uncertainty endogenously generated by chaotic dynamics. Therefore, firms deal with profit uncertainty that is generated, on one side, by uncertainty in the parameters of the price function and, on the other side, by inaccuracy in forecasting the next-period production of the competitor. Then, ambiguity aversion and a constant expectation scheme introduce chaotic dynamics which, in turn, bring further profit uncertainty. All in all, chaotic dynamics are caused by, and at the same time amplify, profit uncertainty.

The road map of the paper is the following. In Section \ref{Sec::MS} we introduce the model setup and we underline preliminary properties. In Section \ref{Sec::PR} we recap the global stability of the Cournot-Nash equilibrium in the complete-information version of the duopoly game. In Section \ref{Sec::ND} we investigate the global dynamics of the duopoly game in case of uncertainty, ambiguity aversion, and constant expectations. Moreover, we discuss the main economic insights of the nonlinear dynamics generated by the model. In Section \ref{Sec::C} we conclude by discussing possible extensions of the current modeling framework. \ref{Appendix:Proof} recaps technical results on decoupled square two-dimensional discrete-time dynamical systems and contains some proofs.

\section{Model setup and preliminary results}\label{Sec::MS}

A Cournot duopoly game with unknown values of payoff function's parameters and uncertainty-averse players, that generalizes the complete-information framework in \cite{SinghVives1984}, is proposed in \cite{CrespiRadiRocca2020}. Firms (or players) produce potentially differentiated goods and each firm produces one type of output only. Let $x$ be the quantity of commodity 1 produced by firm 1 and let $y$ be the quantity of commodity 2 produced by firm 2. The game is played at any time $t\in\mathbb{N}$.

Focusing on a symmetric setting, the inverse demand functions for commodity 1 and commodity 2 are assumed to be given by 
\begin{equation}
P_{1}\left(x,y;b,\gamma\right)=\max\left\{a-bx-\gamma y,0\right\}\quad \text{ and } \quad P_{2}\left(x,y;b,\gamma\right)=P_{1}\left(y,x;b,\gamma\right)
\end{equation}
respectively, where $a>0$ is the choke price and without loss of generality is fixed equal to $1$ in the following,\footnote{In the dynamic model that is developed below, $a$ is a scaling parameter, that is, $a$ is a parameter that impacts only on the quantitative dynamics (amplitude of oscillations) but not on the qualitative dynamics (bifurcation structures). Therefore, the results that follow also hold for any other positive value of $a$ other than $1$.} $b>0$ is the price sensitivity of commodity $i$ to the level of production of firm $i$ (sensitivity of price to quantity), with $i=1,2$, while $\gamma = \sigma b \left(n-1\right)$. Here, $\sigma\in\left[0,1\right]$ measures the degree of substitutability between the two products, $n$ indicates the number of firms in the industry and $\left(n-1\right)y$ represents for firm 1, respectively $\left(n-1\right)x$ represents for firm 2,  the level of production of the rest of the industry in case of (expected) homogeneous competitors.\footnote{Alternatively, $\gamma$ can be interpreted as $ b \alpha$, where the unknown value of $\alpha$ captures the uncertainty about the level of production of the competitor.}

For the sake of comparison, we consider a benchmark case without uncertainty given by the complete-information version of the game where $\sigma \in\left(0,1\right)$ and $n=2$. In the full-fledged version of the model, by contrast, we assume that firms know the choke price but are not aware of the value of $b$, of the degree of product substitutability and there is the (theoretical) possibility that other competitors enter and exit the market at any time. To be precise, firms play duopoly but, at the same time, they do not exclude that other competitors can enter and exit the market at any time.

Concerning the costs of production, they are known to firms and are set equal to zero without loss of generality.\footnote{The generalization to (symmetric) constant marginal costs does not impact on the construction of the best reply functions, therefore on the quantity dynamics of our duopoly. Indeed, fixed costs do not impact on the production decision. Moreover, if prices are negative a firm does not produce and, therefore, marginal costs are zero regardless of their functional form. Finally, if prices are positive, then a positive constant marginal cost of production can be seen as a lower value of the choke price $a$, which is set equal to $1$ being a scaling parameter. By analogous considerations linear marginal costs can be seen as shocks to $b$ and $\gamma$.} It follows that the payoff functions (or profits) of firm $1$ and firm $2$ are given by
\begin{equation}\label{ProfFunc}
\Pi_{1}\left(x,y;b,\gamma\right) = P_{1}\left(x,y;b,\gamma\right)x  \quad \text{ and } \quad \Pi_{2}\left(x,y;b,\gamma\right) = \Pi_{1}\left(y,x;b,\gamma\right)
\end{equation}
respectively.

To summarize, a firm does not know the sensitivity of price to quantity, the exact number of competitors and the degree of substitutability between the two products present in the market. Specifically, both players are neither aware of the (past) realizations of $\gamma$ and $b$ nor of all their possible values, but they rely on their own consistent and conservative belief about the true set of possible realizations $U^{*}$, which is given by an uncertainty set $U$.\footnote{An uncertainty set $U$ is a consistent belief when its worst-case realizations cannot be excluded by firms based on their historical observations and on their information set. Moreover, an uncertainty set $U$ is a conservative belief when its worst-case realizations ensure a maximum-guaranteed payoff which is never higher than the one consistent with $U^{*}$.} The assumption that firms do not know these values of the parameters requires that $b$ and $\gamma$ change over time (are random) as it is the case, for example, when either consumers or their tastes change over time.

Dealing with ambiguity aversion, the worst-case realizations are the only elements of an uncertainty set that impact on firms' decisions. Therefore, we focus on an uncertainty set made of only worst-case realizations and we assume that
\begin{equation}\label{Uibar}
U=\left\{ \left(\overline{b},\underline{\gamma}\right)\text{,} \left(\underline{b},\overline{\gamma}\right) \right\}
\end{equation}
where $\overline{b}\geq \overline{\gamma}\geq \underline{b}\geq \underline{\gamma}$. The restrictions $\overline{b}\geq \underline{\gamma}$ and $\overline{\gamma}\geq \underline{b}$ ensure that the worst-case realization $\left(\overline{b},\underline{\gamma}\right)$ is in the region $b\geq \gamma$ and is consistent with a duopoly game, while the worst-case realization $\left(\underline{b},\overline{\gamma}\right)$ is in the region $b\leq \gamma$ and is consistent with an oligopoly game with more than two players.\footnote{If the firm plays a duopoly game, then $n=2$ and $\gamma=\sigma b\leq b$. If the firm plays an oligopoly game, then $n>2$ and $\gamma=\sigma\left(n-1\right)$ can even be larger than $b$.} Thus, according to $U$ a firm does not know if she plays a duopoly game or a more general oligopoly game. The further conditions $\overline{b}\geq \overline{\gamma}$ and $\underline{b}\geq \underline{\gamma}$ are imposed to capture a sort of predisposition effect towards a duopoly game, which is justified by the fact that a firm always observes the level of production of only one competitor.\footnote{The restrictions $\overline{b}\geq \overline{\gamma}$ and $\underline{b}\geq \underline{\gamma}$ can also be considered as technical assumptions that allow us to focus on a region of the parameter space of the game where there is either a globally stable Cournot-Nash equilibrium or a chaotic regime.}

The true set $U^{*}$ of the possible realizations of $b$ and $\gamma$ does not impact on the configuration of the duopoly game with ambiguity aversion, except that $U$ must be consistent with it, and it could also be made of only a single worst-case realization. $U$ is consistent with $U^{*}$ even in this last case. Indeed, suppose a unique worst-case realization $\left(b^{*},\gamma^{*}\right)$ independently of $\left(x,y\right)$. Suppose that realization $\left(b^{*},\gamma^{*}\right)$ took place in the past (at some generic time $t$). A firm overlooks (at posterior) that $\left(b^{*},\gamma^{*}\right)$ took place as it only observes the market price $p\left(t\right)$, its own level of production $ x\left(t\right)$, the one of the competitor $y\left(t\right)$ and knows that the inverse demand function is piecewise linear and downward slopping. Then, a firm cannot exclude (at posterior) that either $\left(\overline{b},\underline{\gamma}\right)$ or $\left(\underline{b},\overline{\gamma}\right)$, where $\overline{b}>b^{*}>\underline{b}$ and $\overline{\gamma}>\gamma^{*}>\underline{\gamma}$, took place as long as they are consistent with the price equation ($p\left(t\right)=\max\left\{a-b^{*} x\left(t\right)-\gamma^{*} y\left(t\right); 0\right\}$). In other words, with the posterior information that a firm has, it cannot exclude an uncertainty set as $U$.\footnote{Despite being consistent with their historical observations, in this modeling framework we avoid assuming that firms consider the extreme values $\left(\overline{b},0\right)$ and $\left(0,\overline{\gamma}\right)$ to be the two worst cases as they may be perceived unrealistic realizations (like snow in the Sahara). Rather, we discuss how extreme the (assumed) worst-case realizations $\left(\overline{b},\underline{\gamma}\right)$ and $\left(\underline{b},\overline{\gamma}\right)$, with $\overline{b}\geq \overline{\gamma}\geq \underline{b}\geq \underline{\gamma}$, need to be in order to have specific dynamics, such as chaotic attractors.}

Affected by ambiguity aversion, firms do not rely on a probability model of the possible realizations of the unknown values of the parameters $b$ and $\gamma$. By contrast, they maximize their worst-case payoffs consistently with their belief about the uncertainty set. The function that provides the level of production that ensures the maximum-guaranteed payoff of a player, given the production of the competitor, is denoted as worst-case best-reply function. In our duopoly game, the worst-case best-reply function of firm $2$ is given by
\begin{equation}\label{wostcaseBRgen}
f\left(x\right) =
{\arg\max}_{y\geq 0}\left[ {\arg\min}_{\left(b,\gamma\right)\in U} \Pi_{2}\left(x,y;b,\gamma\right)\right]
\end{equation}
Moreover, by symmetry the worst-case best-reply function of firm $1$ is $f\left(y\right)$. 

As underlined in \cite{CrespiRadiRocca2020}, the equilibria of the current duopoly game with ambiguity aversion, denoted Cournot-robust-optimization equilibria, are Cournot-Nash equilibria of the duopoly game without uncertainty and payoff functions equal to the worst-case payoff functions. Therefore, $\left(x^{*},y^{*}\right)$ is a Cournot-robust-optimization equilibrium if and only if $x^{*}=f\left(y^{*}\right)$ and $y^{*}=f\left(x^{*}\right)$.

As opposed to \cite{CrespiRadiRocca2020}, who analyze the static configuration of the game, we assume that firms have constant expectations about the level of production of the competitor, that is $x^{e}\left(t+1\right)=x\left(t\right)$ and $y^{e}\left(t+1\right)=y\left(t\right)$. A player with constant expectations is an agent that decides his next-period level of production overlooking that, as a consequence of his current level of production, the action of the competitor will change. Thus, the assumption of constant expectations introduces a form of naivety well-known in decision theory and behavior economics, see, e.g., \cite{Hammond1976} and \cite{Machina1989}, and supported by experimental evidences, see, e.g., \cite{HeyLotito2009}.

Under this form of naivety, firms update their production levels for the next period according to the following Cournot process
\begin{equation}\label{DST}
\left(x^{\prime},y^{\prime}\right) = 
T\left(x,y\right)
\end{equation}
where
\begin{equation}\label{DST2}
T\left(x,y\right)=\left(f\left(y\right),f\left(x\right)\right)
\end{equation}
and $^{\prime}$ is the unit time advance operator.

The map $T$ hides a second form of naivety which is endogenously generated. Indeed, relaxing the perfect foresight hypothesis, at any time $t\in\mathbb{N}$ firms are worst-case maximizers in the quantity-expectation space, i.e., with respect to $x^{e}\left(t\right)$ and $y^{e}\left(t\right)$, but it is not guaranteed that they are worst-case maximizers in the action space, i.e., with respect to $x\left(t\right)$ and $y\left(t\right)$. Moreover, the map $T$ has the peculiarity that a Cournot-Nash equilibrium of the duopoly game is a fixed point of the map $T$ and a fixed point of the map $T$ is a Cournot-Nash equilibrium of the duopoly game. Since confusion does not arise, in the following the terminology fixed point will be used.

The map $T$ that defines the quantity dynamics of the duopoly game with ambiguity aversion, a trajectory of which represents a Cournot t\^{a}tonnement, has the further peculiarity that its second iterate $T^{2}$ has separate variables given by $\left(x,y\right)\rightarrow \left(f^{2}\left(x\right),f^{2}\left(y\right)\right)$. Therefore, $T$ represents a so-called decoupled square system\footnote{In the sense that square map $T^{2}$ (the second iterate of $T$) generates a one-dimensional decoupled map (specifically, two identical one-dimensional decoupled maps).} and its properties are already studied in the literature, see \cite{BischiMammanaGardini2000} and \cite{TramontanaGardiniPuu2010}. The more general of these properties are taylored to the scope of the current paper and summarized in the following proposition. The more technical ones, related to the existence and stability of periodic cycles, are reported in \ref{Appendix:Proof} for the sake of completeness.

\begin{proposition}\label{MainPropT}
Consider the two-dimensional map $T$ defined in \eqref{DST}-\eqref{DST2} and the one-dimensional map (worst-case best-reply function) $f$ defined in \eqref{wostcaseBRgen}. The map $T$ is such that:
\begin{itemize}
\item[(a)] The diagonal $\Delta $ (the straight line $x=y$) is a trapping set (i.e. $T\left(\Delta\right)\subseteq \Delta )$ and a forward trajectory
\begin{equation}
\left\{\left(x_{0},x_{0}\right),\left(x_{1},x_{1}\right),\ldots,\left(x_{n},x_{n}\right),\left(x_{n+1},x_{n+1}\right),\ldots\right\}\in\Delta
\end{equation}
of $T$ is such that $x_{n+1}=f\left(x_{n}\right)$ for all $n\geq0$;
\medskip
\item[(b)] Let $S$ be the symmetric operator such that $S\left(x,y\right)=\left(y,x\right)$. If $\mathcal{A}$ is an invariant set (as, e.g., cycles, stable/unstable sets and basins of attraction) of the phase plane (i.e. such that $T\left(\mathcal{A}\right)=\mathcal{A}$), so is $S\left(\mathcal{A}\right)$. That is, any invariant set $\mathcal{A}$ either is symmetric with respect to $\Delta$, or the symmetric one (with respect to $\Delta$) also exists.
\item[(c)] If $\left\{x_{1},x_{2},\ldots,x_{N}\right\}$ is the set of all the periodic points of $f$, then the Cartesian product $\left\{x_{1},x_{2},\ldots,x_{N}\right\} \times\left\{x_{1},x_{2},\ldots,x_{N}\right\}$ is the set of all the periodic points of $T$. The converse is also true. Moreover, if $H$ is an interval filled with periodic points of $f$, then the Cartesian product $H\times H$ is filled with periodic points of $T$, if $H$ is an
interval filled with periodic points of $f$ and $\left\{y_{1},y_{2},\ldots,y_{n}\right\}$ a cycle of $f$ not belonging to $H$, then the Cartesian products $H\times\left\{y_{1},y_{2},\ldots,y_{n}\right\}$ and $\left\{y_{1},y_{2},\ldots,y_{n}\right\}\times H$ are filled with periodic points of $T$.
\end{itemize}
\end{proposition}

\medskip

\begin{proof}[Proof of Proposition \ref{MainPropT}]
The proof of (a) is straightforward: let $\left(x,x\right)\in\Delta$ then $T\left(x,x\right)=\left(f\left(x\right),f\left(x\right)\right)\in \Delta$. To prove (b), note that $S\left(T\left(x,y\right)\right)=T\left(S\left(x,y\right)\right)$ follows from the definition of $T$. Now let $\mathcal{A}$ be an invariant set of $T$, so that $T\left(\mathcal{A}\right)=\mathcal{A}$, then $S\left(T\left(\mathcal{A}\right)\right)=T\left(S\left(\mathcal{A}\right)\right)=S\left(\mathcal{A}\right)$ holds. It follows that either $S\left(\mathcal{A}\right)=\mathcal{A}$ (i.e. $\mathcal{A}$ is invariant and symmetric with respect to $\Delta$) or $\mathcal{A}^{\prime }=S\left(\mathcal{A}\right)$ is invariant (being $\mathcal{A}^{\prime }=T\left(\mathcal{A}^{\prime }\right)$). To prove (c), note that $T^{2}\left(x,y\right)=\left(f^{2}\left(x\right),f^{2}\left(y\right)\right)$, and if $n$ is even we also have $T^{n}\left(x,y\right)=\left(f^{n}\left(x\right),f^{n}\left(y\right) \right)$ while if $n$ is odd we have $T^{2n}\left(x,y\right)=\left(f^{2n}\left(x\right),f^{2n}\left(y\right)\right)$. So, let $x$ be a periodic point of $f$ of period $n_{1}$, $y$ be a periodic point of $f$ of period $n_{2}$, and let $n$ be the least common multiple of $n_{1}$ and $n_{2}$, then $\left(x,y\right)$ is a periodic point of $T$ since either $T^{n}\left(x,y\right)=\left(x,y\right)$ (if $n$ is even) or $T^{2n}\left(
x,y\right) =\left(x,y\right)$ (if $n$ is odd). Vice versa, let $\left(x,y\right)$ be a periodic point of $T$ of period $m$, then $T^{m}\left(x,y\right)=\left(x,y\right)$,
and if $m$ is even we also have $T^{m}\left(x,y\right)=\left(f^{m}\left(x\right),f^{m}\left(y\right)\right)=\left(x,y\right)$ while if $m$ is odd we also have $T^{2m}\left(x,y\right)  =\left(f^{2m}\left(x\right),f^{2m}\left(y\right)\right)=\left(x,y\right)$, in both cases leading to $x$ and $y$ periodic points of $f$. The reasoning does not change when the periodic points of $f$ cover densely an interval. This proves point (c).
\end{proof}

\medskip

This proposition (in particular property (a)) indicates that, assuming the same initial level of production for the two firms, that is $x_{0}=y_{0}$, the quantity dynamics of the duopoly game is related to the dynamics of the one-dimensional map $x^{\prime }=f\left(x\right)$. Moreover, this proposition (in particular property (b)) indicates that for any non-symmetric fixed point (or $n$-cycles, or $n$-cyclical chaotic attractors) of $T$, there is another one which is symmetric to it, where symmetry is meant with respect to the diagonal $\Delta$. Moreover, the basins of attraction of two symmetric fixed points (or $n$-cycles, or $n$-cyclical chaotic attractors) of $T$ are also symmetric, where symmetry is again to be understood with respect to the diagonal $\Delta$. The point (b) also implies that, considering the trajectory associated with an initial condition $\left(x_{0},y_{0}\right)$, the symmetric initial condition $\left(y_{0},x_{0}\right)$ is associated to a symmetric (with respect to $\Delta$) trajectory. In addition, Proposition \ref{MainPropT} (in particular property (c)) indicates that the dynamics in $\Delta$ reveals important aspects of the dynamics of $T$ outside $\Delta$.\footnote{Breaking the symmetry of the game, $\Delta$ is not an invariant set anymore but $T$ remains a decoupled square map. Hence, the global dynamics of the model is investigated by studying the global dynamics of a one-dimensional map $x^{'}=f_{1}\circ f_{2}\left(x\right)$ (which is topologically conjugate to $y^{'}=f_{2}\circ f_{1}\left(y\right)$), where $f_{1}$ and $f_{2}$ are the best-reply functions of firms 1 and 2, respectively. See, \cite{BischiMammanaGardini2000}.}

In the following we show that map $T$ generates chaotic dynamics for certain configurations of the uncertainty set.\footnote{A configuration of the uncertainty set $U$ is a set of conditions imposed on the four parameters $\overline{b}$, $\underline{b}$, $\overline{\gamma}$ and $\underline{\gamma}$, in addition to $\overline{b}\geq\overline{\gamma}\geq\underline{b}\geq\underline{\gamma}$.} This impacts on ambiguity aversion of firms, which they are worst-case maximizers in the quantity-expectation space but not so in the action space. At the same time, the duopoly game with ambiguity aversion has three fixed points. For the same values of the parameters, a complete-information duopoly game admits a unique and globally stable fixed point.

\section{The dynamics of the duopoly game without uncertainty}\label{Sec::PR} 

The modeling framework introduced in the previous section generalizes the complete-information setting, where players know the values of the parameters that define their own payoff function. Specifically, for $b=\overline{b}=\underline{b}$ and $\gamma = \overline{\gamma}=\underline{\gamma}$, the set $U$ is a singleton and a complete-information duopoly game, as the one considered in \cite{SinghVives1984}, is recovered. In this quantity-setting duopoly with no uncertainty, $b\geq \gamma \geq 0$ and the best-reply function $f$ is non-increasing and is given by
\begin{equation}\label{BRnom}
f\left(x\right) = \left\{ 
\begin{array}{lcr}
\frac{1}{2b}-\frac{\gamma}{2b}x & \quad  \text{if} \quad\quad  & x_{M}> x \geq 0 \\ 
\\
0 & \quad  \text{if} \quad\quad  & x\geq x_{M}
\end{array}
\right .
\end{equation}
where $x_{M} =  1/\gamma$. 

Assuming the same initial level of production for the two firms, that is $x_{0}=y_{0}$, we have $x_{t}=y_{t}$, $\forall t>0$, see property (a) in Proposition \ref{MainPropT}, and the dynamics of the duopoly game is given by the one-dimensional map $f$. Specifically, the map $f$ has a unique fixed point $x^{*}_{CN} = \frac{1}{2b+\gamma}$, with $x^{*}_{CN}\in\left(0,x_{M}\right)$, known as the Cournot-Nash equilibrium. Since the slope of the function $f$ in $\left(0,x_{M}\right)$ is equal to $-\frac{\gamma}{2b}$, with $b\geq \gamma$, we see that this equilibrium is also locally asymptotically stable. The global stability follows by noting that $f$ is a linear decreasing function for $x<x_{M}$, with $f\left(\left[0,x_{M}\right)\right)\subseteq \left[0,x_{M}\right)$, and it has a zero flat branch for $x\geq x_{M}$. Therefore, every trajectory starting in $\left[0,x_{M}\right)$ approaches asymptotically the Cournot-Nash equilibrium and every trajectory staring in $\left[x_{M},+\infty\right)$ is mapped in zero in one iteration and then converges at the equilibrium.

Then, since $T$ is a decoupled square two-dimensional system, from Property \ref{PropertyA3} in \ref{Appendix:Proof} it follows that $\left(x^{*}_{CN},x^{*}_{CN}\right)$ is a fixed point of $T$. Its local asymptotic stability follows from Proposition \ref{PropertyA5} in \ref{Appendix:Proof}. Its global stability follows by noting that $T$ is linear in $\mathcal{Z} =\left[0,x_{M}\right)\times\left[0,x_{M}\right)$, $\left(x^{*}_{CN},x^{*}_{CN}\right)\in\mathcal{Z}$ and each point outside $\mathcal{Z}$ is mapped in $\mathcal{Z}$ in a single iteration.

This result is well known in the literature and any general setting with $b\geq \gamma$ could be the complete-information benchmark of the duopoly game with ambiguity aversion. Indeed, the complete-information game with $b=\gamma$ respects the conditions $\overline{b}\geq \overline{\gamma}\geq \underline{b}\geq\underline{\gamma}$ when uncertainty vanishes and it could be considered the natural benchmark case. Nevertheless, these conditions characterize only the worst-case realizations when uncertainty is not vanishing, and the set $U$ can be enriched with non-worst-case realizations without affecting the results that follow. Therefore, the complete-information benchmark can be any other realization $\left(b,\gamma\right)$ that lies in the region $b>\gamma$ and that belongs to an uncertainty set with worst-case realizations as the ones in $U$. In addition, also an oligopoly game with more than two firms could potentially be considered the complete-information benchmark case. Indeed, introducing the further assumption that the previous level of production is observed for only two firms and all firms take these two observations to infer the levels of production of the other competitors, then the quantity dynamics of such an oligopoly game with uncertainty is still determined by the decoupled square two-dimensional map introduced above. This map indeed provides the quantity dynamics of the two firms for which productions are observable, while the levels of production of other competitors are auxiliary variables.\footnote{An auxiliary variable is a function of other state variables and does not affect the dynamics of the system.} In this last case, the current model has to be labeled as a Cournot competition with ambiguity aversion and a more appropriate configuration of the worst-case realizations requires, for example, $\overline{\gamma}\geq \overline{b}$ and $\underline{\gamma}\geq \underline{b}$, in addition to $\overline{b}\geq  \underline{b}$ and $\overline{\gamma}\geq\underline{\gamma}$. In this last case, the complete-information benchmark case is described by a dynamical system with more than two dimensions. Nevertheless, assuming a complete-information benchmark case characterized by homogeneous products and firms with the same initial level of production, the map $f$ in \eqref{BRnom} with $\gamma = b\left(n-1\right)$ describes the quantity dynamics in an oligopoly game with $n$ players. This is the classical Theocharis' setup, see \citet{Theocharis1960}, where the existence of more than three firms implies the instability of the Cournot-Nash equilibrium. Having imposed the conditions $\overline{b}\geq \overline{\gamma}\geq \underline{b}\geq\underline{\gamma}$, we consider our modeling framework as a duopoly game with ambiguity aversion.

\section{The dynamics of the duopoly game with ambiguity aversion}\label{Sec::ND}

Assuming uncertainty both on $b$ and $\gamma$, that is imposing $\overline{b}>\underline{b}$ and $\overline{\gamma}>\underline{\gamma}$ in addition to $\overline{b}\geq \overline{\gamma}\geq \underline{b}\geq\underline{\gamma}$, then the uncertainty set is not a singleton and the worst-case best-reply function is derived in \cite{CrespiRadiRocca2020} and is given by
\begin{equation}\label{wostcaseBR}
f\left(x\right) = \left\{ 
\begin{array}{lcr}
f_{l}\left(x\right)=\frac{1-\underline{\gamma} x}{2\overline{b}} & \quad \text{if}\quad \quad  & 0\leq x\leq x_{l} \\ 
f_{m}\left(x\right) =rx & \quad \text{if}\quad \quad  & x_{l}\leq x\leq x_{u} \\ 
f_{r}\left(x\right)=\frac{1-\overline{\gamma}x}{2\underline{b}} & \quad \text{if}\quad \quad  & x_{u} \leq x\leq x_{m} \\ 
0 & \quad \text{if}\quad \quad  & x\geq x_{m}
\end{array}
\right. 
\end{equation}
where $\left(0,x_{l}\right)$ and $f_{l}$ are the left side and left branch of $f$, $\left(x_{l},x_{u}\right)$ and $f_{m}$ are the middle side and middle branch of $f$, $\left(x_{u},x_{m}\right)$ and $f_{r}$ are the right side and right branch of $f$, and
\begin{equation} \label{qunderlinei}
x_{l}  = \frac{1}{\underline{\gamma } +2\overline{b}r} \text{; } \quad x_{u} = \frac{1}{\overline{\gamma}+2\underline{b}r} \text{; }\quad 
x_{m} = \frac{1}{\overline{\gamma}}     \text{; }\quad r=\frac{\overline{\gamma}  -\underline{\gamma}}{\overline{b} -\underline{b}}  
\end{equation}
Differently from the case of no uncertainty, the best-reply function $f$ is piecewise linear and it is not monotonically decreasing. This impacts on the quantity-dynamics of the duopoly game as underlined in the following, where we show that the map $T$ can exhibit chaotic dynamics when the uncertainty set is not a singleton. Exploiting the properties of the decoupled square map $T$, first we analyze its dynamics in the invariant set $\Delta$ given by the diagonal $x=y$. Then we generalize these results by investigating the dynamics outside the diagonal $\Delta$. Indeed, the properties of (symmetric) decoupled square systems stated in Proposition \ref{MainPropT} and in \ref{Appendix:Proof} indicate that existence and stability/instability of period cycles and of cyclical chaotic invariant sets of $T$ can be derived by studying the dynamics on the diagonal $\Delta$.

\subsection{The dynamics in the diagonal of the duopoly game with ambiguity aversion}

Consider the same initial level of production for the two firms, that is $x_{0}=y_{0}$. Then, the dynamics of the oligopoly game is given by $x^{\prime }=f\left(x\right)$. Note that $f$ is a one-dimensional map, which is continuous, piecewise linear and is the worst-case best-reply function in \eqref{wostcaseBR}-\eqref{qunderlinei}. The following proposition indicates the fixed points of this map.

\begin{proposition}\label{ROEoff}
Assume that $U$ is not a singleton (positive level of uncertainty) and consider $r\left(=\frac{\overline{\gamma}-\underline{\gamma}}{\overline{b}-\underline{b}}\right)\neq 1$, then a fixed point of map $f$ (given in \eqref{wostcaseBR}-\eqref{qunderlinei}) exists and is unique. Specifically, for $r<1$ the unique fixed point is in region $\left(0,x_{l}\right)$, while for $r>1$ the unique fixed point is in region $\left(x_{u},x_{m}\right)$. At $r=1$ a border collision occurs and $\left[x_{l}, x_{u}\right]$ is a segment of (non-isolated) fixed points, and no other fixed point exists.
\end{proposition}

\medskip

\begin{proof}[Proof of Proposition \ref{ROEoff}]
To prove that a fixed point of $f$ exists, note that $f$ (given in \eqref{wostcaseBR}-\eqref{qunderlinei}) is continuous and by definition $f\left(0\right)>0$ and $f\left(x\right)=0$ for $x\geq x_{m}>0$. To prove that the fixed point is unique for $r\neq 1$, note that $f_{l}\left(x_{l}\right)=f_{m}\left(x_{l}\right)=rx_{l}>x_{l}$ if and only if $r>1$ and also $f_{m}\left(x_{u}\right)=f_{r}\left(x_{u}\right)=rx_{u}>x_{u}$ if and only if $r>1$. It follows that the values of the function $f$ at the two kink points $x_{l}$ and $x_{u}$ are either both above (when $r>1$) or both below (when $r<1$) the diagonal. Then, the uniqueness of the fixed point follows by observing that $f_{l}$ and $f_{r}$ are both decreasing functions and the unique fixed point belongs to $\left(x_{u},x_{m}\right)$ when $r>1$ while it belongs to $\left(0,x_{l}\right)$ for $r<1$. For $r=1$, we have $f_{m}\left(x\right)=x$ for all $x\in\left[x_{l},x_{u}\right]$, that is, a segment of fixed points exists. The two extreme points of this segment of fixed points are also kink points. Therefore, at $r=1$ a border collision occurs. To prove that no other fixed point exists, note that $f$ is decreasing in $\left[0, x_{l}\right]$ and in $\left[x_{u}, x_{m}\right]$, with $f\left(x\right)=0$ for $x>x_{m}$.
\end{proof}

\medskip

Before discussing the global dynamics of map $f$, let us consider the local stability properties of its fixed points. 

\begin{proposition}\label{FixedPointsStabilityDelta0}
Consider $x^{*}$ to be a fixed point of $f$ defined in \eqref{wostcaseBR}-\eqref{qunderlinei}, we have that
\begin{itemize}
\item It is locally asymptotically stable when it belongs to $\left(0,x_{l}\right)$, i.e. when it belongs to the left side of $f$;
\item It is marginally stable\footnote{Here, marginally stable means derivative of $f$ equal to $+1$. More generally, a fixed point that is marginally stable is a fixed point that is stable but not attracting. It is stable because all trajectories starting in a sufficiently close neighborhood of the fixed point remain in a neighborhood of the fixed point itself, however, it is not attracting, as it does not attract trajectories starting in a neighborhood of it.} when it belongs to $\left(x_{l},x_{u}\right)$, i.e. if it belongs to the middle side of $f$;
\item If it belongs to $\left(x_{u},x_{m}\right)$, i.e. to the right side of $f$, then it is locally asymptotically stable for $\overline{\gamma }< 2\underline{b}$, it undergoes a degenerate flip bifurcation for $\overline{\gamma }= 2\underline{b}$, and it is repelling for $\overline{\gamma }> 2\underline{b}$.
\end{itemize}
\end{proposition}

\medskip

\begin{proof}[Proof of Proposition \ref{FixedPointsStabilityDelta0}]
The slope of the left branch of $f$ is $-\frac{\underline{\gamma }}{2\overline{b}}$
and $\frac{\underline{\gamma }}{2\overline{b}}\in\left(0,1\right)$ since $\overline{b}>\underline{b}\geq\underline{\gamma }$. Hence, since the left branch of $f$ is linear, when a fixed point in the left side of $f$ exists, it is always attracting. This proves the first point of the proposition. The slope of the middle branch of $f$ is $r$ and $x^{*}$ belongs to the middle side $\left(x_{l},x_{u}\right)$ if and only if $r=1$. Then, the marginal stability of $x^{*}$ follows by the linearity of the middle branch of $f$. This proves the second point of the proposition. The slope of the right branch of $f$ is $-\frac{\overline{\gamma }}{2\underline{b}}$. Since the right branch of $f$ is linear and we can have $\overline{\gamma }\gtreqless 2\underline{b}$, it follows that a fixed point in the right branch is attracting for $\overline{\gamma }<2\underline{b}$, it undergoes a degenerate flip bifurcation at $\overline{\gamma }=2\underline{b}$ (see \cite{SushkoGardini2010} for an overview on degenerate bifurcations) and it is repelling otherwise. This proves the third point of the proposition.
\end{proof}

\medskip

The conditions of existence and the local stability properties of the possible fixed points of the map $f$ stated in Propositions \ref{ROEoff} and \ref{FixedPointsStabilityDelta0} are used to characterize the global dynamics of $f$. Specifically, the following theorem identifies six dynamic configurations that are represented in Figure \ref{Fig::qualitative}.

\begin{figure}[hbt!]
	\begin{centering}
                   \includegraphics[scale=0.55]{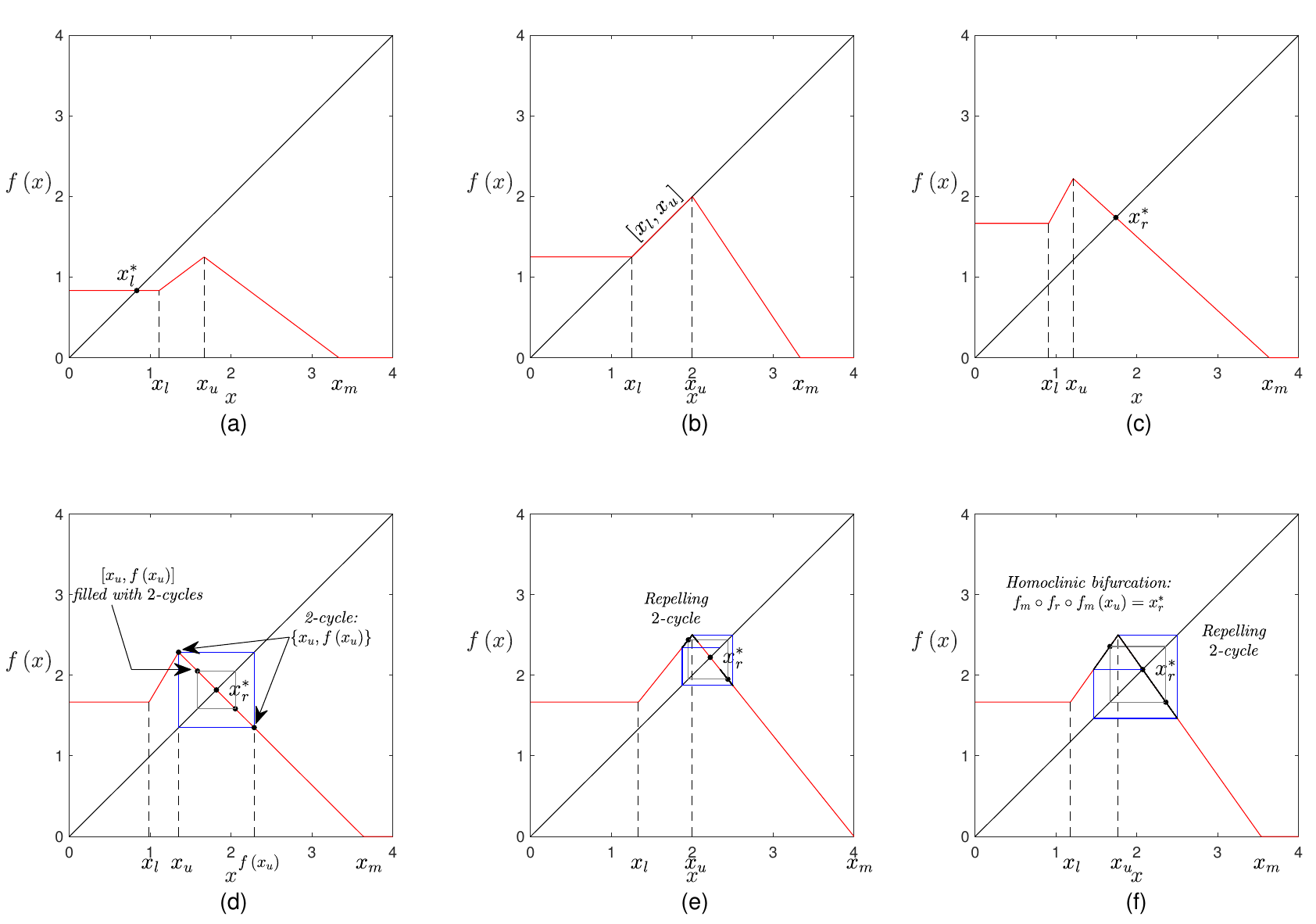}
			\caption{Graph of $f$ (in red). In (a) the case (i) in Theorem \ref{Th:GDdelta0}, parameters $a=1$, $\overline{b}=0.6$, $\underline{b}=0.2$, $\overline{\gamma}=0.3$ and $\underline{\gamma}=0$. In (b), the case (ii) in Theorem \ref{Th:GDdelta0}, parameters $a=1$, $\overline{b}=0.4$, $\underline{b}=0.1$, $\overline{\gamma}=0.3$ and $\underline{\gamma}=0$. In (c), the case (iii)-(a) in Theorem \ref{Th:GDdelta0}, parameters $a=1$, $\overline{b}=0.3$, $\underline{b}=0.15$, $\overline{\gamma}=0.275$ and $\underline{\gamma}=0$. In (d), the case (iii)-(b) in Theorem \ref{Th:GDdelta0}, parameters $a=1$, $\overline{b}=0.3$, $\underline{b}=0.1375$, $\overline{\gamma}=0.275$ and $\underline{\gamma}=0$. In (e), the case (iii)-(c) in Theorem \ref{Th:GDdelta0}, parameters $a=1$, $\overline{b}=0.3$, $\underline{b}=0.1$, $\overline{\gamma}=0.25$ and $\underline{\gamma}=0$. In (f), the case (iii)-(d) in Theorem \ref{Th:GDdelta0} (homoclinic bifurcation), parameters $a=1$, $\overline{b}=0.3$, $\underline{b}=0.1$, $\overline{\gamma}\approx 0.28284275$ and $\underline{\gamma}=0$. (For interpretation of the references to color in this figure caption, the reader is referred to the web version of this paper.)}\label{Fig::qualitative}
	\end{centering}
\end{figure}

\begin{theorem}\label{Th:GDdelta0}
Consider map $f$ in \eqref{wostcaseBR}-\eqref{qunderlinei}. We have the following cases:
\begin{itemize}
\item[(i)] For $\left(\overline{\gamma }-\underline{\gamma }\right)<\left(\overline{b}-
\underline{b}\right)$, that is $r<1$, the fixed point in the left side of $f$, that is $x_{l}^{\ast }=\frac{1}{\underline{\gamma }+2\overline{b}}$, is globally
attracting;
\item[(ii)] For $\left(\overline{\gamma }-\underline{\gamma }\right)=\left(\overline{b}-
\underline{b}\right)$, that is $r=1$, the segment $\left[x_{l},x_{u}\right]$ is filled with
fixed points, stable but not attracting. Any point either is fixed or it is
mapped into a fixed point in a few iterations;
\item[(iii)] For $\left(\overline{\gamma }-\underline{\gamma }\right)>\left(\overline{b}-\underline{b}\right)$, that is $r>1$, the unique fixed point of $f$, that is $x_{r}^{\ast }=\frac{1}{\overline{\gamma }+2\underline{b}}$, is in the right side and we have the following sub-cases:
\begin{itemize}
\item[(a)] $\overline{\gamma }<2\underline{b}$, then $x_{r}^{\ast }$ is globally attracting;
\item[(b)] $\overline{\gamma }=2\underline{b}$, then $x_{r}^{\ast }$ undergoes a degenerate flip bifurcation and the segment $\left[x_{u},f\left(x_{u}\right)\right]$ is filled with 2-cycles, stable but not attracting, and any non-fixed point is either 2-periodic or it is mapped into a 2-cycle in a few iterations;
\item[(c)] $\overline{\gamma }>2\underline{b}$, then $x_{r}^{\ast }$ is repelling, a repelling 2-cycle $\left\{x_{1},x_{2}\right\}$ of $f\left(x\right)$ exists and, close to the bifurcation, the unique attracting set consists in $2^{k}$-cyclical chaotic intervals, where $k\geq 0$ depends on the value $r$.
\item[(d)] For $\overline{\gamma }>2\underline{b}$ and $f^{2}\left(x_{u}\right)>x_{l}$, at $f_{m}\circ f_{r}\circ f_{m}\left(x_{u}\right)=x_{r}^{\ast }$ the first homoclinic bifurcation of the repelling fixed point in the right branch causes the transition from two to one chaotic interval. For $f_{m}\circ f_{r}\circ f_{m}\left(x_{u}\right)<x_{r}^{\ast }$ there is one unique chaotic interval (globally attracting).
\end{itemize}
\end{itemize}
\end{theorem}

\medskip

\begin{proof}[Proof of Theorem \ref{Th:GDdelta0}]
Consider (i). Then $x_{l}^{\ast }=\frac{1}{\underline{\gamma }+2\overline{b}}$ is the unique fixed point of $f$ and lies in the left side of $f$, see Proposion \ref{ROEoff}. Moreover, it is globally attracting since the graph of map $f$ is below the diagonal for $x>x_{l}$. See Figure \ref{Fig::qualitative}(a). This proves (i). Consider (ii). By Proposion \ref{ROEoff} the segment $\left[x_{l},x_{u}\right]$ is filled with fixed points. See Figure \ref{Fig::qualitative}(b). As stated in Proposition \ref{FixedPointsStabilityDelta0}, these fixed points are stable but not attracting (marginally stable). In this case, since the map is decreasing on both sides of the segment $\left[x_{l},x_{u}\right]$ any point on the left side of $x_{l}$ is mapped on the right side of $x_{l}$ and any point on the right side of $x_{u}$ is mapped on the left side of $x_{u}$. It follows that any point of the phase space is either fixed or mapped into a fixed point in a few iterations. This proves (ii). Consider (iii)-(a). Then, $x_{r}^{\ast}$ is the unique fixed point of $f$ and is locally asymptotically stable, see Propositions \ref{FixedPointsStabilityDelta0} and \ref{ROEoff}. To prove the global stability of $x_{r}^{\ast}$ let us proceed as follows. Note that any point $x>x_{r}^{\ast}$ is mapped in $\left[0,x_{r}^{\ast}\right)$ so that we can consider the points in $\left[0,x_{r}^{\ast}\right)$. Moreover, when $x_{r}^{\ast}$ is locally asymptotically stable, it is $f\left(x_{u}\right)<x_{m}$, and we have that $f$ is a contraction in $\left[x_{u},f\left(x_{u}\right)\right]$. In particular, $f^{2}\left(x_{u}\right)>x_{u}$ and in the interval $\left[x_{u},x^{\ast}\right)$ $f^{2}$ is monotone increasing with $f^{2}\left(x\right)>x$. Since $f$ is increasing in the middle side we have $x_{l}<f\left(x_{l}\right)<f\left(x_{u}\right)$ and thus in the interval $\left[x_{l},x_{u}\right]$ it is
$f^{2}\left(x_{l}\right)>x_{l}$, $f^{2}\left(x_{u}\right)>x_{u}$ and the second iterate $f^{2}$ connecting $f^{2}\left(x_{l}\right)$ and $f^{2}\left(x_{u}\right)$ is either only decreasing (if $f\left(x_{l}\right)\geq x_{u}$) or first increasing and then decreasing (if
$f\left(x_{l}\right)<x_{u}$), but always above the diagonal. In the interval $\left[0,x_{l}\right]$ for $\underline{\gamma}=0$ it is $f\left(\left[0,x_{l}\right]\right)=f\left(x_{l}\right)>x_{l}$ and we are done, while for $\underline{\gamma}>0$ we have always $f\left(0\right)<x_{m}$ and thus $f^{2}\left(0\right)>0$, and, as above, $f^{2}\left(x_{l}\right)>x_{l}$, so that the second iterate $f^{2}$ connecting $f^{2}\left(0\right)$ and $f^{2}\left(x_{l}\right)$ is either always decreasing (if $f\left(x_{l}\right)\geq x_{u})$ or first increasing and then decreasing (if $f\left(x_{l}\right)<x_{u}$) but always above the diagonal. This proves that the second iterate of the map, that is $f^{2}$, is above the diagonal in the interval $\left[0,x_{r}^{\ast}\right)$ which implies the non-existence of 2-cycles of $f$. This is a well-known condition (stability of a fixed point and absence of $2$-cycles, see \cite{ElaydiSacker2004}) leading to the global attractivity of $x_{r}^{\ast}$. This proves (iii)-(a). Consider (iii)-(b). The slope of $f_{r}$ is equal to $-1$. Then the fixed point $x_{r}^{\ast}$ undergoes a degenerate flip bifurcation (since $f_{r}$ is linear), so that the segment $\left[x_{u},f\left(x_{u}\right)\right]$ is filled with 2-cycles, stable but not attracting. See Figure \ref{Fig::qualitative}(d). The graph of $f^{2}$ in the interval $\left[x_{u},f\left(x_{u}\right)\right]$ is on the diagonal (i.e. $f^{2}\left(x\right)=x$), while in the interval $\left[0,x_{u}\right)$ the graph of $f^{2}$ is above the diagonal (the reasonings are the same as for the case (iii)-(a)), so that outside $\left[x_{u},f\left(x_{u}\right)\right]$ a 2-cycle cannot exist and the interval $\left[x_{u},f\left(x_{u}\right)\right]$ is globally attracting. That is, any point inside this interval, different from the fixed point, is 2-periodic and any point outside it is mapped into that interval in a finite number of iterations. This proves (iii)-(b). Consider (iii)-(c). Note that the slope of the middle branch of $f$ is greater than $+1$, the slope of $f$ in the right branch is smaller than $-1$, and the map in this invariant interval is a skew-tent map. From \cite{SushkoAvrutinGardini2015} and \cite{AvrutinGardiniSushkoTramontana2019}, it follows that the degenerate flip bifurcation, that occurs when $\overline{\gamma}=2\underline{b}$ (case (iii)-(b)), leads to an unstable 2-cycle and to a unique attracting set that consists in $2^{k}$ chaotic intervals for some $k\geq 0$ (the value of $k$ depending on the value of $r>1$). Consider (iii)-(d). Note that $f\left(x_{l}\right)<x_{r}^{*}$, which can be rewritten as $\frac{r}{\underline{\gamma}+2\overline{b}r}<\frac{1}{\overline{\gamma}+2\underline{b}}$, is equivalent to $\overline{\gamma}r+2\underline{b}r<\underline{\gamma}+2\overline{b}r$. Since $\overline{b}\geq \overline{\gamma}$ always and $\overline{\gamma}>2\underline{b}$ as we consider case (iii)-(d), it follows that $\overline{\gamma}r+2\underline{b}r<2\overline{\gamma}r<\underline{\gamma}+2\overline{b}r$, i.e. $f\left(x_{l}\right)<x_{r}^{*}$ is ensured. Then, $\left[f^{2}\left(x_{u}\right),f\left(x_{u}\right)\right]$ is an absorbing interval for $f$ when $f^{2}\left(x_{u}\right)>x_{l}$, and $f$ in the absorbing interval $\left[f^{2}\left(x_{u}\right),f\left(x_{u}\right)\right]$ is a skew-tent map. This ensures that the condition given on the third iterate of $x_{u}$ leads to the homoclinic bifurcation of the fixed point and transition to a unique chaotic interval, see again \cite{SushkoAvrutinGardini2015}.
\end{proof}

\medskip

As specified in Theorem \ref{Th:GDdelta0}, a conservative approach to uncertainty can generate chaotic dynamics in a simple duopoly populated by players with constant expectations. In particular, we observe that chaotic dynamics require a specific configuration of the uncertainty set. First of all, comparing the two worst-case realizations, the spread on the values of the parameter $\gamma$, that defines the competitive interaction between the two firms, must be higher than the spread on the values of the parameter $b$, that represents the price reactivity to one's own level of production.\footnote{Note that the spread on the values of $b$ at the worst-case realizations, given by $\overline{b}-\underline{b}$, and the spread on the values of $\gamma$ at the worst-case realizations, given by $\overline{\gamma}-\underline{\gamma}$, cannot be used to compare the uncertainty about the values of the two parameters $b$ and $\gamma$. Rather, their impact on profit uncertainty is relevant and depends on the levels of production of firms.} Second, a firm must not know that it plays a duopoly game and must consider the existence of more than two competitors plausible, which is required by the parameter realization $\left(\overline{\gamma},\underline{b}\right)$ when $\overline{\gamma}>2\underline{b}$.

This last condition is related to Theocharis' result that can be summarized as follows, the stability of the Cournot-Nash equilibrium is guaranteed in an oligopoly game with no more than three firms. Moreover, the Cournot-Nash equilibrium becomes unstable when more than three firms are involved, see \citet{Theocharis1960}.\footnote{Theocharis considers an oligopoly with $n$ firms and homogeneous products, that is  $\gamma = b$. Specifically, the map (best-reply function) is $g\left(x\right)=\frac{1}{2b}-\frac{n-1}{2}x$ if $x\leq \bar{x}=1/\left(b\left(n-1\right)\right)$ and $g\left(x\right)=0$ for $x>\bar{x}$. In case of constant expectations and $n\geq 4$, each $x$ is mapped in a few iterations in the 2-cycle $\left\{0,\bar{x}\right\}$. Theocharis' result is generalized in \cite{Fisher1961}, where it is shown that an oligopoly with more than three firms can be stable if the marginal costs are not constant, see also \cite{Hahn1962}. Moreover, it has been recently discussed in the framework of evolutionary oligopoly games in \cite{BischiLamantiaRadi2015} and \cite{HommesOcheaTuinstra2018}.} Analogously to the current setting, Theocharis considers oligopoly games with constant expectations, constant marginal costs, inverse demand functions as the one here considered, and the particular case of homogeneous products. In this respect, our results add to the literature on the stability of the Cournot-Nash equilibrium in oligopoly games by showing that by introducing uncertainty and a worst-case approach, even in a duopoly game the Cournot-Nash equilibrium can lose stability. In order to occur, it is required that the degree of substitutability of products is uncertain (that justifies $\left(\overline{b},\underline{\gamma }\right)$ with $\overline{b}>\underline{\gamma }$) and uncertainty has to regard the composition of the industry as well, with firms that consider possible the entry of at least two further competitors into play (that justifies $\left(\underline{b},\overline{\gamma }\right)$ with $\underline{b}<2\overline{\gamma }$). 

It is worth remarking that in the current duopoly game with uncertainty, chaotic attractors emerge when the Cournot-Nash equilibrium loses stability, while chaotic dynamics cannot occur in the oligopoly games considered by Theocharis. This difference in the out-of-equilibrium dynamics is due to the worst-case best-reply functions of our duopoly game with ambiguity aversion, which are not monotonically decreasing as the best-reply functions of the oligopoly games considered by Theocharis. The configuration of the worst-case best-reply functions in the different cases identified by Theorem \ref{Th:GDdelta0} are observable in Figure \ref{Fig::qualitative}.
 
The conditions for the existence of chaotic attractors provided in Theorem \ref{Th:GDdelta0} reveal a further interesting aspect. Specifically, a reduction of profit uncertainty (generated by the uncertainty about the values of $b$ and $\gamma$) does not imply more stability. To understand this point, first note that chaotic attractors are (and their existence require equilibria to be) in the regions of the action space where $\left(\underline{b},\overline{\gamma}\right)$ is a worst-case realization. In such a context, a higher $\underline{b}$ implies a further reduction of the profit at the worst-case realization, increasing therefore the profit uncertainty measured as the profit gap between the best-case realization $\left(\overline{b},\underline{\gamma}\right)$ and the worst-case realization $\left(\underline{b},\overline{\gamma}\right)$.\footnote{In line with the definition of $U$, a non-worst-case realization is a best-case realization.} On the other hand, for $r<1$ we cannot have chaotic dynamics for the duopoly game with ambiguity aversion and its unique equilibrium is in a region of the action space where the worst-case realization is $\left(\overline{b},\underline{\gamma}\right)$. In such a context, the profit at the best-case realization (in this case $\left(\underline{b},\overline{\gamma}\right)$) reduces by increasing $\underline{b}$, therefore less profit uncertainty.

Then, the one-dimensional bifurcation diagram in Figure \ref{Fig::1DLaura}(a) underlines that increasing the value of $\underline{b}$ in the parameter region $r<1$ (i.e. $\underline{b}<\overline{b}-\overline{\gamma}+\underline{\gamma}=0.1$), the unique Cournot-Nash equilibrium of the game persists but, as the profit uncertainty vanishes because $\underline{b}$ approaches the value $\overline{b}-\overline{\gamma}+\underline{\gamma} = 0.1$ (at which $r=1$), the Cournot-Nash equilibrium loses stability (entering the region of the action space where both $\left(\underline{b},\overline{\gamma}\right)$ and $\left(\overline{b},\underline{\gamma}\right)$ are worst-case realizations) and chaotic dynamics appear by increasing $\underline{b}$ further. For $\underline{b}>\overline{b}-\overline{\gamma}+\underline{\gamma} = 0.1$, chaotic dynamics persist and profit uncertainty increases by increasing $\underline{b}$ (a chaotic trajectory crosses both regions of the action space where $\left(\underline{b},\overline{\gamma}\right)$ is the only worst case realization and regions of the action space where $\left(\underline{b},\overline{\gamma}\right)$ and $\left(\overline{b},\underline{\gamma}\right)$ are worst-case realizations at the same time). However, further increasing $\underline{b}$ (therefore further increasing profit uncertainty) we have that a Cournot-Nash equilibrium regains stability. It happens for $\underline{b}>\frac{1}{2}\overline{\gamma}$. To summarize, chaotic dynamics are observed when profit uncertainty (generated by the uncertainty about the values of the parameters $b$ and $\gamma$) is limited in the region of the action space where chaotic dynamics are. This confirms that it is not the level of profit uncertainty that impacts on the stability of an equilibrium of our duopoly game with ambiguity aversion, but the configuration of the uncertainty set.

\begin{figure}[hbt!]
	\begin{centering}
		\includegraphics[scale=0.55]{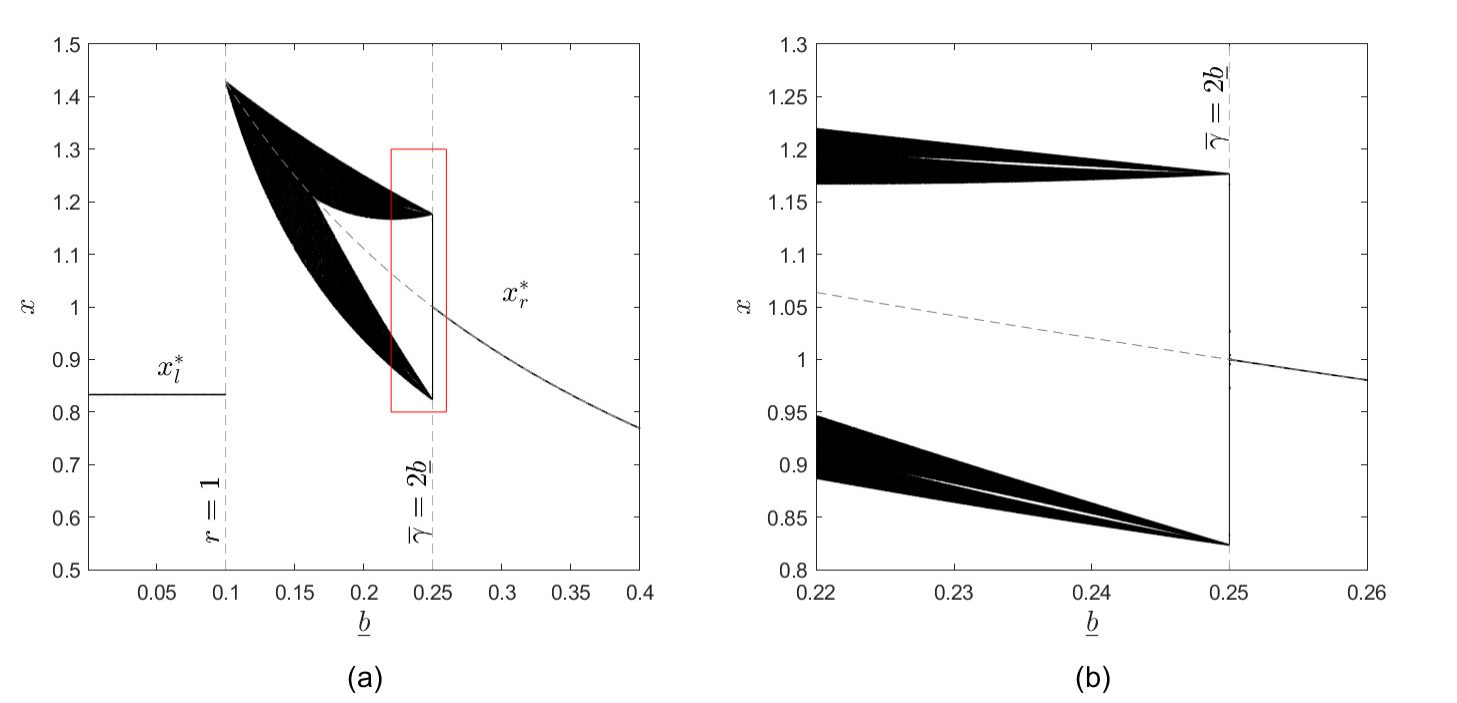}
			\caption{In (a) one-dimensional bifurcation diagram showing the long-run dynamics of $x$ as a function of $\underline{b}$ varying in the range $\left[0,0.4\right]$. In (b) an enlargement around the bifurcation value $\overline{\gamma}=2\underline{b}$. Parameters:  $\overline{b}=0.6$, $\overline{\gamma}=0.5$ and $\underline{\gamma}=0.0$. (For interpretation of the references to color in this figure caption, the reader is referred to the web version of this paper.)}\label{Fig::1DLaura}
	\end{centering}
\end{figure}

Specifically, chaotic dynamics imply a configuration of the uncertainty set such that $\overline{\gamma}>2\underline{b}$, that is a firm is uncertain between a duopoly game and an oligopoly game with more than two competitors, and $\overline{\gamma}-\underline{\gamma}>\overline{b}-\underline{b}$. The reader may wonder that a firm could learn that it plays a duopoly game, if it is indeed the case. This is a plausible observation in a context of perfect expectations. On the contrary, in a framework of constant expectations it is difficult to disentangle forecasting errors from parameter uncertainty, especially in a context of chaotic dynamics. Specifically, chaotic dynamics are caused by, and are at the same time are a source of, uncertainty. They are caused by profit uncertainty because without it the duopoly game does not exhibit chaotic dynamics. They are a source of profit uncertainty as they contribute to generating forecasting errors which do not allow firms to be worst-case maximizers in the action space, see Figure \ref{Fig::TS}(a) for an example. This implies a worst-case profit in the action space which is different from the maximum-guaranteed payoff in the action space and which is different from both the maximum-guaranteed payoff and the best-possible payoff that are expected, see Figure \ref{Fig::TS}(b). It follows that we can have more profit uncertainty in a chaotic regime even though the payoff uncertainty generated by $b$ and $\gamma$ might be higher at an equilibrium solution.\footnote{Note that a chaotic attractor in $\Delta$ involves both the middle side (where both $\left(\overline{b},\underline{\gamma}\right)$ and $\left(\underline{b},\overline{\gamma}\right)$ are the worst-case realizations) and the right side of $f$ (where only $\left(\underline{b},\overline{\gamma}\right)$ is the worst-case realization). Instead, the Cournot-Nash equilibrium is on the right side of $f$. Then, profit uncertainty generated by $b$ and $\gamma$, and measured as profit gap at the two realizations of $U$, might be higher at the Cournot-Nash equilibrium and is observed to be so. Note that profit uncertainty generated by $b$ and $\gamma$ depends indeed on the levels of production of firms.} To better underline these aspects, note that firms have the same levels of production as we consider the dynamics along the diagonal $\Delta$. Then let us consider the guaranteed achievable profit (based on constant expectations) for firm 1 at time $t$, which is given by $\Pi^{wc}_{1}\left(f\left(x^{e}\left(t\right)\right),x\left(t\right)\right)$, where $\Pi_{1}^{wc}$ indicates that it is the worst-case payoff computed consistently with the arguments. Hence, this is the worst-case payoff in the action space, that is, with regards to the realized levels of production which are given by $x\left(t\right)$ for both firms since $f\left(x^{e}\left(t\right)\right)=x\left(t\right)$ being $x^{e}\left(t\right)$ the constant expectation scheme. Moreover, let us consider for firm 1 the \emph{maximum-guaranteed expected payoff}, given by $\Pi^{wc}_{1}\left(f\left(x^{e}\left(t\right)\right),x^{e}\left(t\right)\right)$, and the \emph{best-possible expected payoff}, given by $\Pi^{bc}_{1}\left(f\left(x^{e}\left(t\right)\right),x^{e}\left(t\right)\right)$. These are the maximum-guaranteed payoff and best-possible payoff in the quantity-expectation space. Then, looking at the time series generated by a chaotic regime, we observe that firms are not worst-case maximizers in the action space, see Figure \ref{Fig::TS}(a), and we observe that the guaranteed achievable profits are often lower than then maximum-guaranteed expected payoffs, see Figure \ref{Fig::TS}(b). This confirms that firms are subject to a further form of profit uncertainty which is not the one derived by the uncertainty about $b$ and $\gamma$.

Another relevant aspect is worth underlining. A chaotic dynamics generates a historical volatility in the realized profits that it is not observable at a Cournot-Nash equilibrium. Indeed, by adopting a worst-case approach, the historical profit volatility that is observable at an equilibrium is due to the realizations of $b$ and $\gamma$ that may change over time. In addition to that, in a chaotic regime the profit volatility is inflated by the levels of production of firms which are also volatile and by the generated profits which are often lower than the maximum-guaranteed expected payoffs. This is a further source of risk that may impact on firms' investment decisions, such as investments in \emph{R\&D} or in process innovation. The current model is too simple to consider these aspects which are however worth investigating in the future by introducing ambiguity aversion and by relaxing the perfect foresight hypothesis in more sophisticated models of industrial organization.

\begin{figure}[hbt!]
	\begin{centering}
		\includegraphics[scale=0.55]{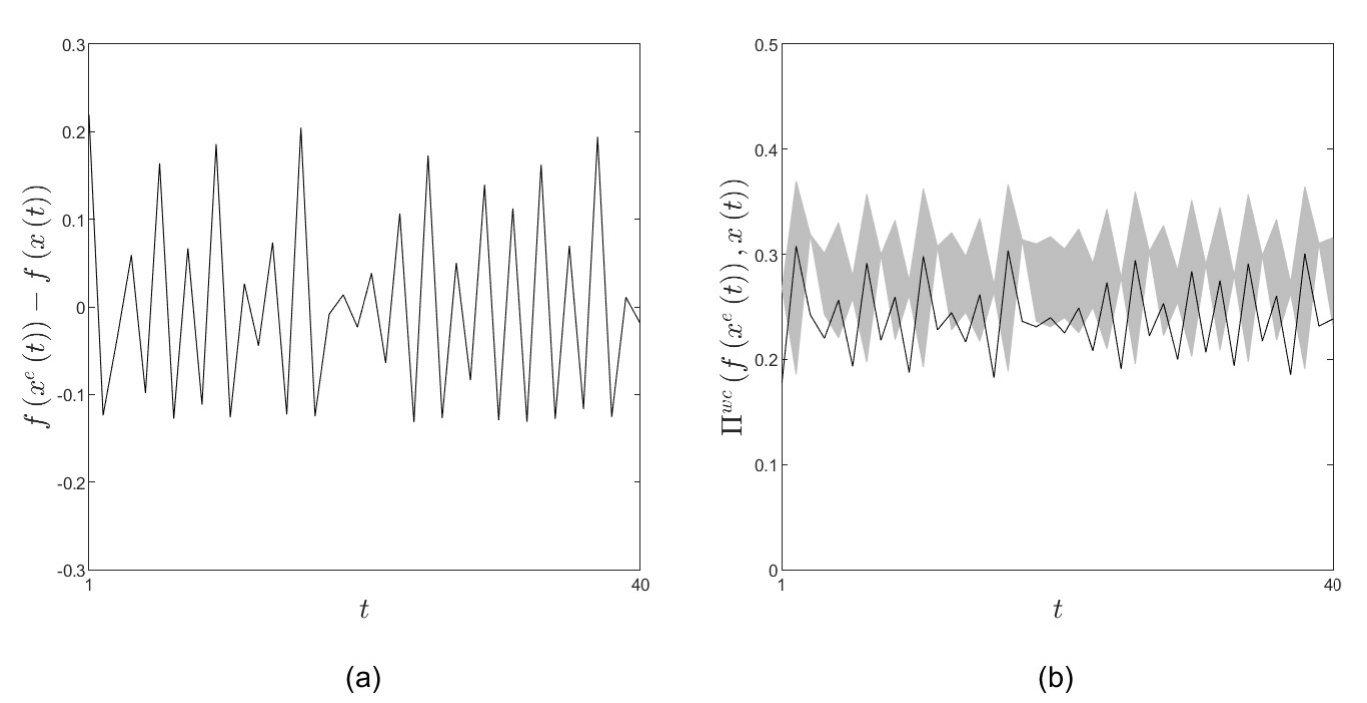}
			\caption{In (a) time series showing the long-run dynamics of the worst-case best response to $x^{e}\left(t\right)$ (ambiguity-averse firm in the quantity-expectation space) minus the worst-case best response to $x\left(t\right)$ (ambiguity-averse firm in the action space). In (b) time series of the worst-case achievable profits based on constant expectations (solid line in black). Moreover, for each time $t$ all the profits between the maximum-guaranteed expected payoff, that is $\Pi_{1}^{wc}\left(f\left(x^{e}\left(t\right)\right),x^{e}\left(t\right)\right)$, and best-possible expected payoff (consistently with $U$), that is $\Pi_{1}^{bc}\left(f\left(x^{e}\left(t\right)\right),x^{e}\left(t\right)\right)$, are in gray. Parameters:  $\overline{b}=0.6$, $\underline{b}= 0.15$, $\overline{\gamma}=0.5$ and $\underline{\gamma}=0.0$. (For interpretation of the references to color in this figure caption, the reader is referred to the web version of this paper.)}\label{Fig::TS}
	\end{centering}
\end{figure}

As a further observation, the reader may also argue that constant expectations are critical to justify in the case of periodic dynamics as firms can learn that these expectations are systematically wrong. However, periodic dynamics are not generated by the duopoly game with ambiguity aversion and predictive learning based on historical observations is limited in a chaotic regime. Indeed, a further element that is observable by analyzing the bifurcation diagrams of Figure \ref{Fig::1DLaura}, is the abrupt transition from a stable Cournot-Nash equilibrium to a chaotic dynamics, which is obtained by perturbing the parameter $\underline{b}$. This is typical of piecewise-smooth dynamical systems.

\subsection{The dynamics outside the diagonal of the duopoly game with ambiguity aversion}

Consider two different initial levels of production for the two firms that populate the duopoly game with ambiguity aversion, that is, consider the dynamics of $T$ outside the diagonal. To study the dynamics of $T$ outside the diagonal, let us start by underlying the following relations between $f$ and $T$.

\begin{corollary}\label{ThDafAT}
Consider map $f$ in \eqref{wostcaseBR}-\eqref{qunderlinei} and map $T$ in \eqref{DST}-\eqref{DST2}.
\begin{itemize}
\item[(j)] If $x^{\ast}$ is a repelling (resp. attracting) fixed point of $f$, then $\left(x^{\ast},x^{\ast}\right)$ is a repelling (resp. attracting) fixed point of $T$ belonging to the diagonal $\Delta$;
\item[(jj)] If $f$ has a 2-cycle $\left\{x_{1},x_{2}\right\}$, then $T$ has two fixed points outside the diagonal $\Delta$ given by $\left(x_{1},x_{2}\right)$ and $\left(x_{2},x_{1}\right)$, and a 2-cycle of $T$ on $\Delta$ given by $\left\{
\left(x_{1},x_{1}\right),\left(x_{2},x_{2}\right)\right\}$. Moreover, for $r>1$ and $\overline{\gamma}>2\underline{b}$, $f$ has an unstable 2-cycle $\left\{x_{1},x_{2}\right\}$ which leads to a pair of repelling fixed points of $T$ outside the diagonal $\Delta$ given by $\left(x_{1},x_{2}\right)$ and $\left(x_{2},x_{1}\right)$, and to a repelling 2-cycle of $T$ on $\Delta$ given by $\left\{
\left(x_{1},x_{1}\right),\left(x_{2},x_{2}\right)\right\}$;
\item[(jjj)] If $x_{1}^{\ast}$ and $x_{2}^{\ast}$ are two different fixed points of $f$, then $\left(x_{1}^{\ast},x_{1}^{\ast}\right)$ and $\left(x_{2}^{\ast},x_{2}^{\ast}\right)$ are two fixed points of $T$ on the diagonal $\Delta$ and $\left\{\left(x_{1}^{\ast},x_{2}^{\ast}\right),\left(x_{2}^{\ast},x_{1}^{\ast}\right)\right\}$ is a $2$-cycle of $T$ external to the diagonal.
\end{itemize}
\end{corollary}

\medskip

\begin{proof}[Proof of Corollary \ref{ThDafAT}]
The results (j) and (jj) are immediate consequences of the properties recalled in \ref{Appendix:Proof} and of Theorem \ref{Th:GDdelta0}. In fact, the fixed points of $T$ are related to fixed points and 2-cycles of $f$. That is, $\left(x^{\ast },y^{\ast }\right)$ is a fixed point of $T$ if and only if $T\left(x^{\ast },y^{\ast }\right)=\left(f\left(y^{\ast }\right),f\left(x^{\ast }\right)\right)=\left(x^{\ast },y^{\ast}\right)$. This also leads to $T^{2}\left(x^{\ast },y^{\ast }\right)=\left(f^{2}\left(x^{\ast}\right),f^{2}\left(y^{\ast }\right)\right)=\left(x^{\ast },y^{\ast }\right)$. It follows that if $x^{\ast }$ is a fixed point of $f$, then $\left(x^{\ast },x^{\ast }\right)$ is a fixed point of $T$ belonging to the diagonal $\Delta$. Moreover, any 2-cycle $\left\{x_{1},x_{2}\right\}$ of $f$ leads to a pair of fixed points of $T$ outside the diagonal $\Delta$: $\left(x_{1},x_{2}\right)$ and $\left(x_{2},x_{1}\right)$, and to a 2-cycle of $T$ on $\Delta$: $\left\{\left(x_{1},x_{1}\right),\left(x_{2},x_{2}\right)\right\}$. In \ref{Appendix:Proof}, it is also recalled the stability of these cycles of $T$ which is related to that of the fixed points and 2-cycles of $f$. Property (jjj) follows by noting that $x_{1}^{\ast}$ and $x_{2}^{\ast}$ fixed points of $f$ implies $T^{2}\left(x_{1}^{\ast},x_{2}^{\ast}\right) = \left(f^{2}\left(x_{1}^{\ast}\right),f^{2}\left(x_{2}^{\ast}\right)\right)$, $T^{2}\left(x_{2}^{\ast},x_{1}^{\ast}\right) = \left(f^{2}\left(x_{2}^{\ast}\right),f^{2}\left(x_{1}^{\ast}\right)\right)$, $T^{2}\left(x_{1}^{\ast},x_{2}^{\ast}\right)=\left(x_{1}^{\ast},x_{2}^{\ast}\right)$ and $T^{2}\left(x_{2}^{\ast},x_{1}^{\ast}\right)=\left(x_{2}^{\ast},x_{1}^{\ast}\right)$.
\end{proof}

\medskip

\begin{figure}[hbt!]
	\begin{centering}
		\includegraphics[scale=0.55]{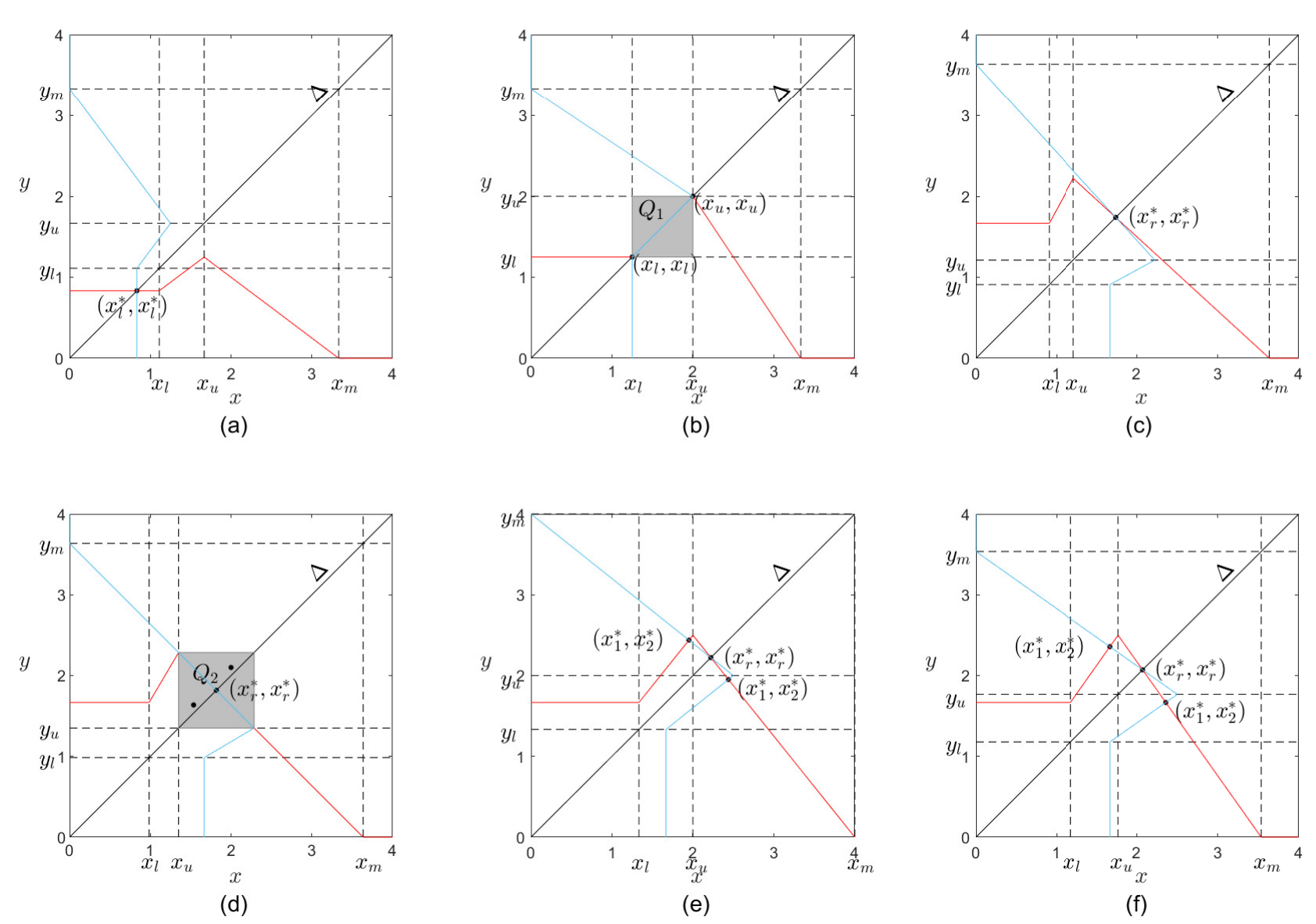}
			\caption{Worst-case best-reply function of firm 1 ($f\left(y\right)$ in light blue) and worst-case best-reply function of firm 2 ($f\left(x\right)$ in red). Parameters as in Figure \ref{Fig::qualitative}. (For interpretation of the references to color in this figure caption, the reader is referred to the web version of this paper.)}\label{Fig::qualitative2d}
	\end{centering}
\end{figure}

Combining the dynamics of map $f$ stated in Theorem \ref{Th:GDdelta0} with the relations between $f$ and $T$ stated in Corollary \ref{ThDafAT}, we can describe the attracting sets of system $T$. In particular, in the cases (i) and (iii)-(a) of Theorem \ref{Th:GDdelta0}, map $T$ has a unique fixed point on the diagonal, respectively $\left(x_{l}^{\ast},x_{l}^{\ast }\right)$ and $\left(x_{r}^{\ast },x_{r}^{\ast }\right)$, which is globally attracting, see Figure \ref{Fig::qualitative2d}(a),(c).

The two sets of parameter values that identify the cases (ii) and (iii)-(b) of Theorem \ref{Th:GDdelta0} are both bifurcation values for $T$. Specifically, for parameters as in case (ii), map $T$ has the rectangle $Q_{1}=\left[x_{l},x_{u}\right]\times\left[x_{l},x_{u}\right]$ filled with periodic points, given by 2-cycles external to the diagonal $\Delta$ and filled with fixed points on the portion of $\Delta$ that belongs to $Q_{1}$ (from property (jjj) of Corollary \ref{ThDafAT}). See Figure \ref{Fig::qualitative2d}(b). The 2-cycles are all stable but not attracting, and any point of the phase plane either is in $Q_{1}$ or its trajectory enters $Q_{1}$ in a few iterations.
For parameters as in case (iii)-(b) of Theorem \ref{Th:GDdelta0}, map $T$ has the rectangle $Q_{2}=\left[x_{u},f\left(x_{u}\right)\right]\times\left[x_{u},f\left(x_{u}\right)\right]$ filled with periodic points, given by 2-cycles ($\left\{\left(x_{1},x_{1}\right),\left(x_{2},x_{2}\right)\right\}$, $\left\{\left(x_{3},x_{3}\right),\left(x_{4},x_{4}\right)\right\}$, ...) on the portion of the diagonal $\Delta$ that belongs to $Q_{2}$, filled with fixed points ($\left(x_{1},x_{2}\right)$, $\left(x_{2},x_{1}\right)$, $\left(x_{3},x_{4}\right)$, $\left(x_{4},x_{3}\right)$, ...) on the second diagonal of $Q_{2}$, (from property (jj) of Corollary \ref{ThDafAT}) and filled with 2-cycles ($\left\{\left(x_{1},x_{3}\right),\left(x_{4},x_{2}\right)\right\}$, $\left\{\left(x_{1},x_{4}\right),\left(x_{3},x_{2}\right)\right\}$, ...) outside the diagonals of $Q_{2}$. See Figure \ref{Fig::qualitative2d}(d). The 2-cycles are all stable but not attracting, and any point of the phase plane either is in $Q_{2}$ or its trajectory enters $Q_{2}$ in a few iterations.

In the remaining cases (iii)-(c) and (iii)-(d) of Theorem \ref{Th:GDdelta0}, that is, when $f$ exhibits a $2^{k}$-cyclical chaotic attractor, map $T$ has three (unstable) equilibria, see Figure \ref{Fig::qualitative2d}(e),(f). Moreover, it has chaotic attracting sets, consisting in one unique rectangle or in cyclical rectangles, depending on the value of $k=0,1,2\ldots$, where $2^{k}$ indicates the periodicity of the cyclical chaotic attractor of $f$. Specifically, if $f$ has a unique chaotic interval $I=\left[f^{2}\left(x_{u}\right),f\left(x_{u}\right)\right]=\left[f_{r}\left(\frac{r}{\overline{\gamma }+2\underline{b}r}\right),\frac{r}{\overline{\gamma }+2\underline{b}r}\right]$, then $k=0$ and the map $T$ has a unique chaotic set consisting in the square $R=I\times I$, see Figure \ref{Fig::1-chaotic}. If $f$ has 2-cyclical chaotic intervals $I_{1}=\left[f^{3}\left(x_{u}\right),f\left(x_{u}\right)\right]$ and $I_{2}=\left[f^{2}\left(x_{u}\right),f^{4}\left(x_{u}\right)\right]$, then $k=1$ and the map $T$ has a 2-cyclical chaotic attractor intersecting the diagonal $\Delta $, given by $I_{1}\times I_{1}$ and $I_{2}\times I_{2},$ plus two more distinct attracting chaotic rectangles given by $I_{1}\times I_{2}$ and $I_{2}\times I_{1}$ (i.e. two more $1$-cyclical chaotic attractors), see Figure \ref{Fig::2-chaotic}. If $f$ has 4-cyclical chaotic intervals $I_{1}$, $I_{2}$, $I_{3}$ and $I_{4}$, then $k=2$ and the map $T$ has a 4-cyclical chaotic attractor intersecting the diagonal $\Delta$, given by $I_{1}\times I_{1}\cup I_{2}\times I_{2} \cup I_{3}\times I_{3}\cup I_{4}\times I_{4}$, plus three more 4-cyclical chaotic attractors outside $\Delta$, given by $I_{1}\times I_{2}\cup I_{3}\times I_{2} \cup I_{3}\times I_{4}\cup I_{1}\times I_{4}$, $I_{2}\times I_{1}\cup I_{2}\times I_{3} \cup I_{4}\times I_{3}\cup I_{4}\times I_{1}$ and $I_{1}\times I_{3}\cup I_{4}\times I_{2} \cup I_{3}\times I_{1}\cup I_{2}\times I_{4}$. See Figure \ref{Fig::4-chaotic}. In general we can state the following.

\begin{theorem}\label{ChaoticDynamicsT}
Consider map $f$ in \eqref{wostcaseBR}-\eqref{qunderlinei} for $r>1$, $\overline{\gamma}>2\underline{b}$ and $f\left(x_{u}\right)<x_{m}$ (i.e. $\frac{1}{\overline{\gamma}/r+2\underline{b}}<\frac{1}{\bar{\gamma}}$), assuming that map $f$ has $2^{k}$-cyclical chaotic intervals $I_{1},...,I_{2^{k}}$ for $k\geq 0$. Then map $T$ has one $2^{k}$-cyclical chaotic set crossing the diagonal $\Delta $, given by the squares $I_{1}\times I_{1},\ldots,I_{2^{k}}\times I_{2^{k}}$. Moreover, for $k=1$, there are two $1$-cyclical chaotic attractors outside $\Delta$, that is, $I_{1}\times I_{2}$ and $I_{2}\times I_{1}$, while, for $k>1$, map $T$ has $\left(2^{k}-1\right)$ distinct $2^{k}$-cyclical chaotic rectangles external to $\Delta$, consisting of chaotic rectangles $I_{i}\times I_{j}$ for $i\neq j$ and $i,j\in\left\{1,\ldots,2^{k}\right\}$.
\end{theorem}

\medskip

The proof of Theorem \ref{ChaoticDynamicsT} is the same as the proof of Property \ref{PropertyA3} in \ref{Appendix:Proof}, substituting periodic points $x_{i}$ and $x_{j}$ with intervals $I_{i}$ and $I_{j}$.

\begin{figure}[hbt!]
	\begin{centering}
		\includegraphics[scale=0.55]{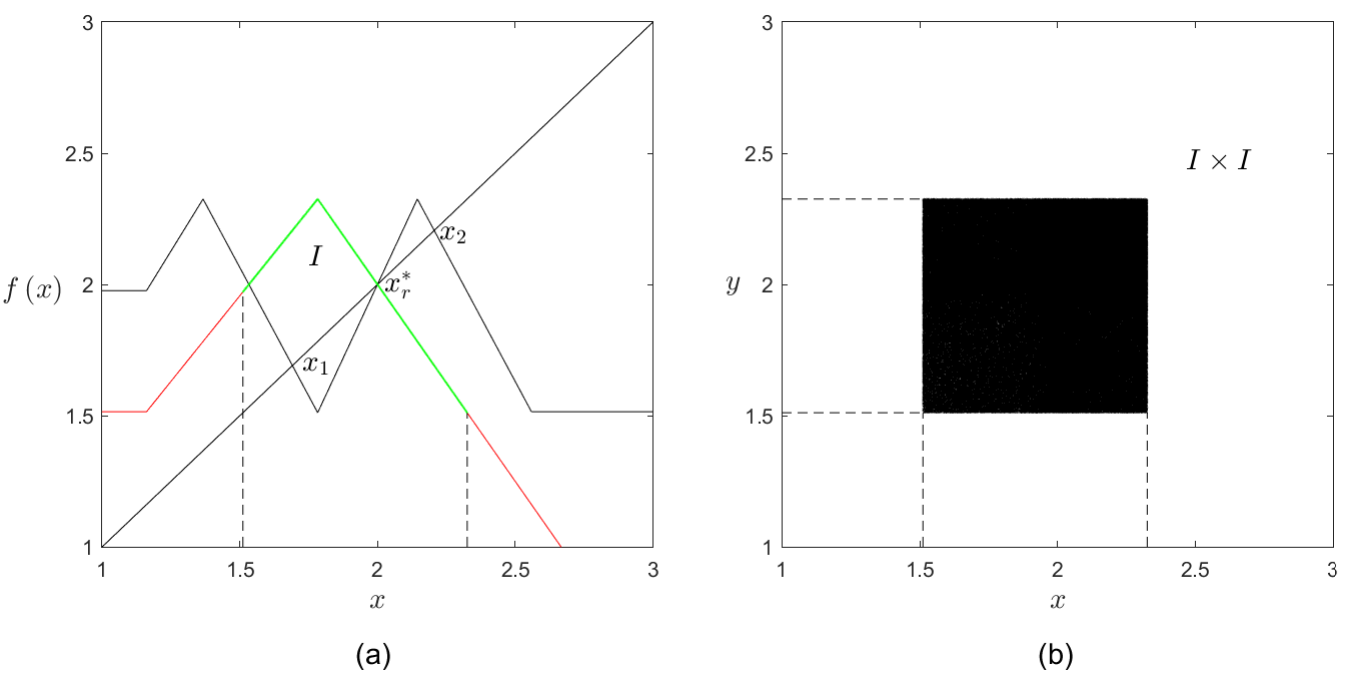}
			\caption{In (a), graph of function $f$ (in red), graph of function $f^{2}$ (in black), chaotic invariant sets $I$ (in green in the graph of $f$). In (b) state space $\left(x,y\right)$ of the map $T$, in white the basin of attraction of the chaotic attractor $I\times I$ (global attractor). Parameters: $\overline{b}=0.33$, $\underline{b}=0.1$, $\overline{\gamma}=0.3$ and $\underline{\gamma}=0$. (For interpretation of the references to color in this figure caption, the reader is referred to the web version of this paper.)}\label{Fig::1-chaotic}
	\end{centering}
\end{figure}

\begin{figure}[hbt!]
	\begin{centering}
		\includegraphics[scale=0.55]{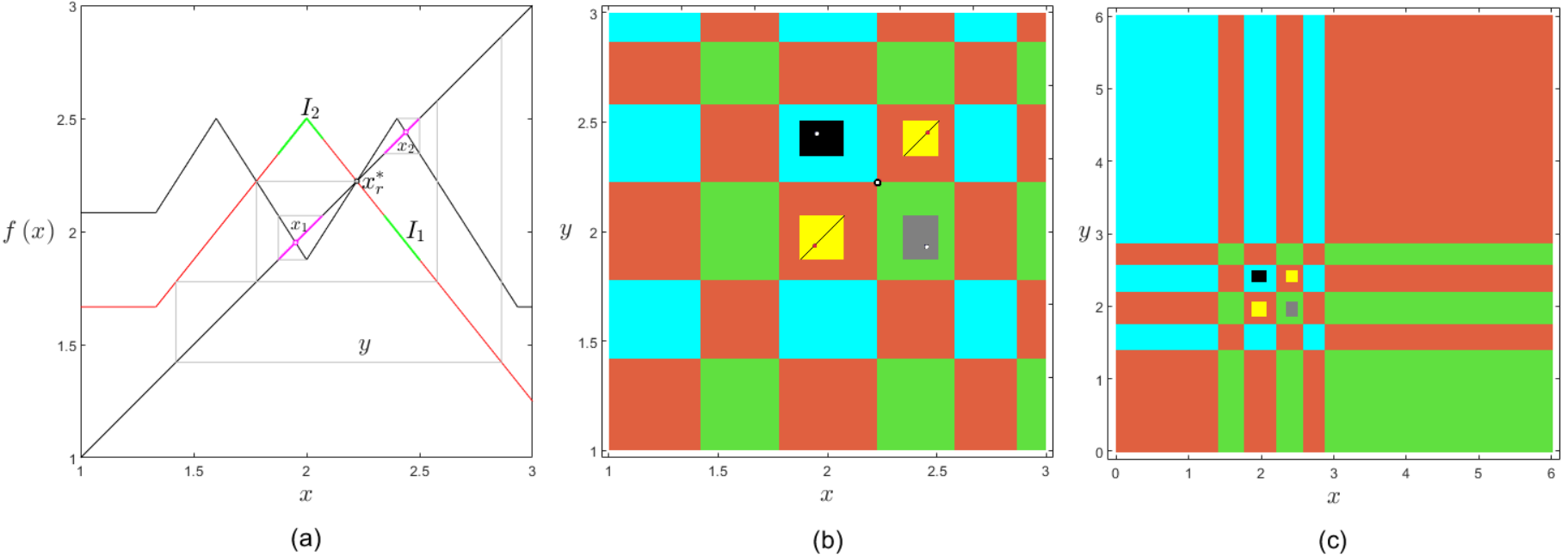}
			\caption{In (a), graph of function $f$ (in red), graph of function $f^{2}$ (in black), chaotic invariant sets $I_{i}$, $i=1,2$, (in green in the graph of $f$ and in magenta in the diagonal). In (b) state space $\left(x,y\right)$ of the map $T$, in red the basin of attraction of the 2-cyclical chaotic attractor $I_{1}\times I_{1}$ and $I_{2}\times I_{2}$ (yellow region), in green the basin of attraction of the chaotic rectangle $I_{1}\times I_{2}$ (gray rectangle), in light blue the basin of attraction of the chaotic rectangle $I_{2}\times I_{1}$ (black region), the two red points in the diagonal $\Delta$ indicate a repelling 2-cycle, in white are the fixed points (robust optimization equilibria). In (c), a larger window of the state space $\left(x,y\right)$ of the map $T$. Parameters: $\overline{b}=0.3$, $\underline{b}=0.1$, $\overline{\gamma}=0.25$ and $\underline{\gamma}=0$. (For interpretation of the references to color in this figure caption, the reader is referred to the web version of this paper.)}\label{Fig::2-chaotic}
	\end{centering}
\end{figure}

\begin{figure}[hbt!]
	\begin{centering}	
		\includegraphics[scale=0.55]{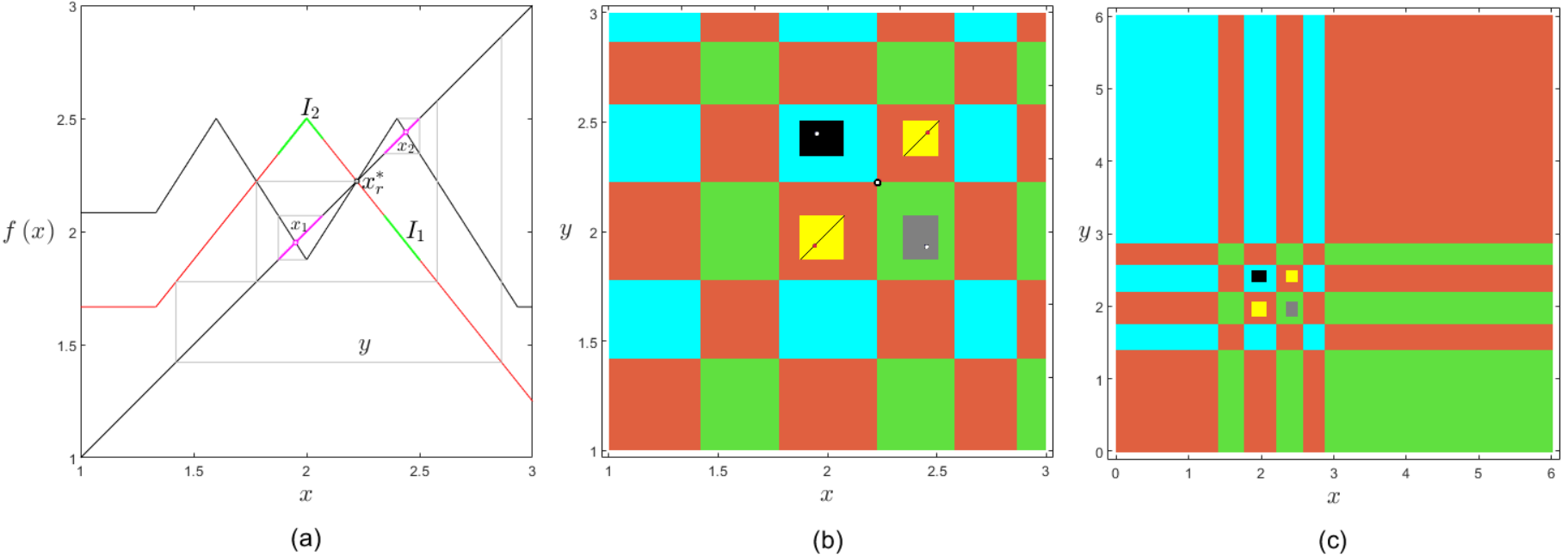}
			\caption{In (a), graph of function $f$ (in red), graph of function $f^{2}$ (in black), chaotic invariant sets $I_{i}$, $i=1,2,3,4$, (in green in the graph of $f$). In (b) state space $\left(x,y\right)$ of the map $T$, in light blue the basin of attraction of the 4-cyclical chaotic attractor depicted in blue, in green the basin of attraction of the 4-cyclical chaotic attractor symmetric to it (with respect to the diagonal $\Delta$) and depicted in dark green, in red the basin of attraction of the 4-cyclical chaotic attractor depicted in yellow, finally, in gray the basin of attraction of the 4-cyclical chaotic attractor depicted in black. In (c) a larger window of the phase space. Parameters: $\overline{b}=0.29$, $\underline{b}=0.11$, $\overline{\gamma}=0.235$ and $\underline{\gamma}=0.01$. (For interpretation of the references to color in this figure caption, the reader is referred to the web version of this paper.)}\label{Fig::4-chaotic}
	\end{centering}
\end{figure}

As stated in Theorem \ref{ChaoticDynamicsT} the existence of chaotic dynamics is associated with either a single attractor or multiple attractors, and some insightful considerations are worth making regarding their basins of attraction. First of all, when there exists one unique chaotic interval $I$ for map $f$, it is globally attracting and also map $T$ has a unique chaotic region, $I\times I$, which is also globally attracting, an example is shown in Figure \ref{Fig::1-chaotic}.\footnote{By an abuse of notation, global attractiveness has to be understood unless sets of zero measures.}

More insightful are the analysis of the cases where $2^{k}$ cyclical chaotic attractors, with $k\geq 1$, exist. Consider the case $k=1$, i.e. of two cyclical chaotic intervals for $f$ shown in Figure \ref{Fig::2-chaotic}(a), and denoted by $I_{1}$ and $I_{2}$. Then map $f$ has a repelling fixed point and a repelling 2-cycle internal to the chaotic intervals, that is inside $I_{1}\cup I_{2}$, whose boundaries are given by the critical point $c_{1}=f\left(x_{u}\right)$ and its three images by $f$, that we denote by $c_{i}$, with $\left(i=2,3,4\right)$. $I_{1}=\left[c_{3},c_{1}\right]$, $I_{2}=f\left(I_{1}\right)=\left[c_{2},c_{4}\right]$. The chaotic intervals $I_{1}$ and $I_{2}$ are invariant for $f^{2}$ and the preimages of any rank of the repelling fixed point $x_{r}^{\ast }$ (rank-1 preimage, two rank-2 preimages and one rank-3 preimage of $x_{r}^{\ast }$) separate intervals of points converging to one or the other chaotic interval for map $f^{2}$, see Figure \ref{Fig::2-chaotic}(a), where the existing preimages are shown, and are thus related to the basins of attraction for map $T$. Specifically, considering the same example but looking at map $T$, this map has three repelling fixed points and a repelling 2-cycle on the diagonal (where the dynamics is provided by map $f$), a 2-cyclical chaotic attractor $I_{1}\times I_{1}$ and $I_{2}\times I_{2}$ whose boundaries are known segments (see the yellow chaotic set in Figure \ref{Fig::2-chaotic}(b),(c), including the repelling 2-cycle, and basin in red), and two symmetric chaotic rectangles $I_{1}\times I_{2}$ and $I_{2}\times I_{1}$ (see the black and gray rectangles in Figure \ref{Fig::2-chaotic}(b),(c), including each a repelling fixed point, and basins of different colors), see Proposition \ref{MainPropT}. The boundaries of the chaotic sets are known segments on horizontal and vertical lines of equation $x=c_{i}$ and $y=c_{i}$ $\left(i=1,2,3,4\right)$. The related basins are bounded by segments belonging to horizontal and vertical lines of equation $x=x_{r}^{\ast }$ and its preimages (see Figure \ref{Fig::2-chaotic}(b),(c)).

The configuration of the basins of attraction becomes even more complicated when $4$-cyclical chaotic attractors exist. See Figure \ref{Fig::4-chaotic}, where in panel (a), close to the degenerate flip bifurcation of the fixed point $x_{r}^{\ast }$, we observe a 4-pieces chaotic set for map $f$. In correspondence to the same parameter configuration, in Figure \ref{Fig::4-chaotic}(b),(c), we observe that map $T$ has a stable 4-cyclical chaotic set crossing the diagonal (with basin in red), coexisting with a stable 4-cyclical chaotic set transverse to the diagonal (with basin in gray), and with a stable 4-cyclical chaotic set above the diagonal (with basin in light blue) and its symmetric one below the diagonal (with basin in green). The basins are separated by the preimages of the repelling fixed point and by the preimages of the repelling 2-cycle (whose periodic points are not included in the 4-cyclical chaotic intervals on the diagonal). 

All in all, the basins of attraction are disconnected whenever chaotic attractors coexist and this represents a further element of complexity that the conservative approach to uncertainty introduces in a duopoly game with constant expectations. It is remarkable to observe that the disconnected basins of attraction underline that a firm with an initially smaller market share can dominate the market in the long run (see the structure of the disconnected basins of attraction in Figures \ref{Fig::2-chaotic} and \ref{Fig::4-chaotic}). However, for this to occur both uncertainty and ambiguity aversion are required as disconnected basins of attraction are not observable in a complete-information duopoly game. In addition, we observe that disconnected basins of attraction, cross regions of the action space characterized by different worst-case realizations of the parameters $b$ and $\gamma$. This indicates a transitory dynamics towards a chaotic attractor where the firm's ambiguity aversion is violated (firms expect a worst-case realization, which does not however occur).

To summarize, when the dynamics of $f$ lead to a fixed, globally attracting point, then also map $T$ has a unique globally attracting fixed point, see Corollary \ref{ThDafAT}. Alternatively, attracting chaotic regions emerge. To better highlight the extent of the parameter region for which chaotic dynamics occur, two-dimensional bifurcation diagrams are drawn. Specifically, Figure \ref{Fig::2d-bif-diag-new} shows that in the $\left(\overline{\gamma},\overline{b}\right)$-parameter plane (at fixed values of the other parameters), the yellow region (where $\left(x_{l}^{\ast },x_{l}^{\ast }\right)$ is globally attracting fixed point) is dominant. Outside the yellow region, we have either an orange region, where the equilibrium $\left(x_{r}^{\ast },x_{r}^{\ast }\right)$ is globally stable, or a white region where chaotic attractors ($2^{n}$-cyclical chaotic attractors) coexist and are attracting. The border of the yellow region with either the orange or the white regions, is the line $r=1$ and indicates the presence of a segment filled with equilibria, see Figure \ref{Fig::qualitative2d}(d). By contrast, the red line is the separator of the orange region and of the white region and is the bifurcation curve $\overline{\gamma}=2\underline{b}$.

The two bifurcation diagrams in Figure \ref{Fig::2d-bif-diag-new} point out that it is sufficient to increase the spread in the values of $b$ at the worst-case realizations, that is $\overline{b}-\underline{b}$ sufficiently large, to have a globally stable Cournot-Nash equilibrium. On the contrary, it is not sufficient to increase the spread on the values of $\gamma$ at the worst-case realizations, that is $\overline{\gamma}-\underline{\gamma}$ sufficiently large, to have an unstable Cournot-Nash equilibrium as the further condition $\overline{\gamma}>2\underline{b}$ is required. Moreover, it is also worth observing that increasing $\underline{b}$ from $0.1$, see panel (a) in Figure \ref{Fig::2d-bif-diag-new}, to $0.15$, see panel (b) in Figure \ref{Fig::2d-bif-diag-new}, the width of the line strip representing chaotic region increases although the spread on the values of $b$ at the worst-case realization decreases. This is counterintuitive and underlines, once more, that not only do the spreads on the uncertain parameters at the worst-case realizations impact on the stability of the Cournot-Nash equilibrium, but also the configuration of the uncertainty set.

As a further remark, we underline that by relaxing the assumption that $\overline{b}>\overline{\gamma}$, chaotic dynamics persist only for $\overline{\gamma}$ slightly larger than $\overline{b}$. Then, further increasing $\overline{\gamma}$, periodic dynamics appear and cycles may have different periodicity. As the scope is to show chaotic dynamics in a simple duopoly setting with ambiguity aversion, the investigation of the dynamics in the parameter region $\overline{b}<\overline{\gamma}$ is omitted.

\begin{figure}[hbt!]
	\begin{centering}
		\includegraphics[scale=0.65]{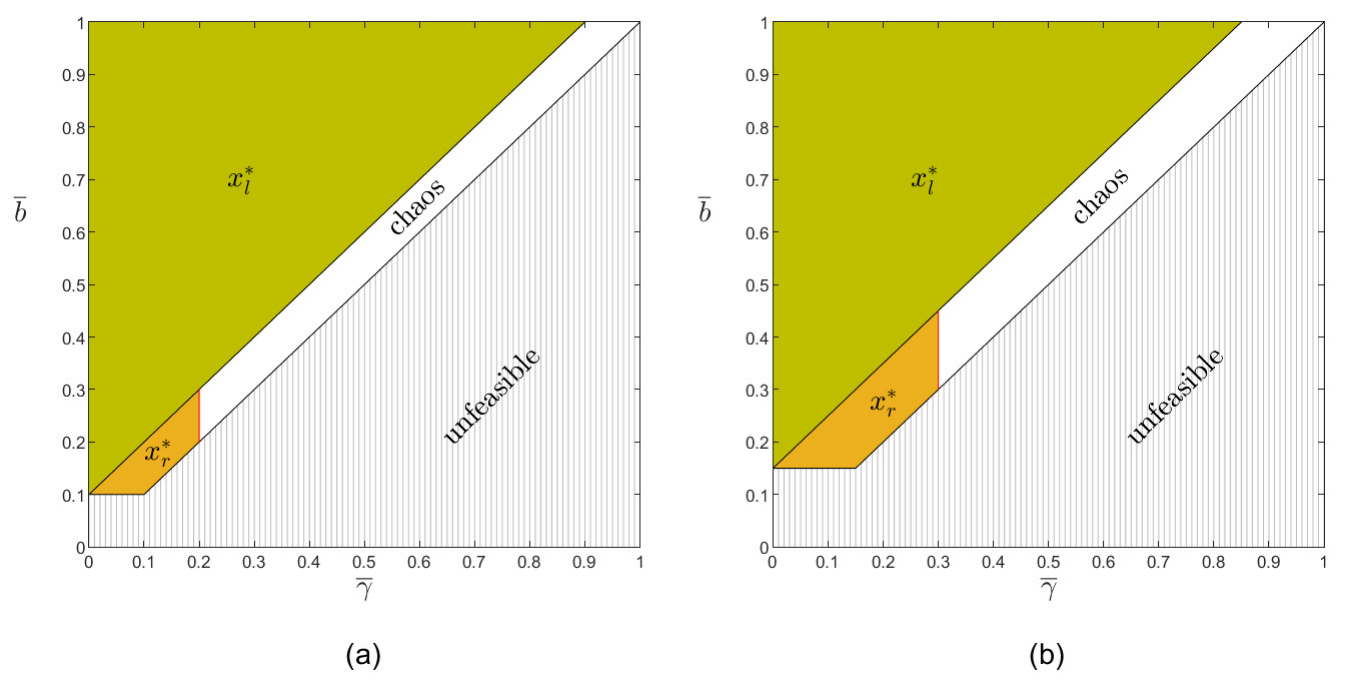}
			\caption{$\left(\overline{\gamma},\overline{b}\right)$-parameter plane (the region of interest is $\overline{b}\geq \overline{\gamma}$ and $\overline{b}\geq \underline{b}$). In yellow the stability region of the equilibrium $\left(x^{*}_{l},x^{*}_{l}\right)$ (labeled as $x^{*}_{l}$ in the figure), in orange the stability region of the equilibrium $\left(x^{*}_{r},x^{*}_{r}\right)$ (labeled as $x^{*}_{r}$ in the figure), in white the stability region of a chaotic regime. In (a), $\underline{b}=0.1$. In (b), $\underline{b}=0.15$. Remaining parameter: $\underline{\gamma}=0$. (For interpretation of the references to color in this figure caption, the reader is referred to the web version of this paper.)}\label{Fig::2d-bif-diag-new}
	\end{centering}
\end{figure}

The final remark regards the robustness of the results. Let us assume that only one firm has constant expectations, say firm 2, while the other firm has perfect foresight expectations, say firm 1. In this case, firm 2 acts naively, firm 1 acts sophisticatedly (being a worst-case maximizer both in the quantity-expectation space and in the action space) and the quantity dynamics of the duopoly is given by $\left(x^{\prime},y^{\prime}\right)=\left(f\circ f\left(x\right),f\left(x\right)\right)$. Then, $y$ is an auxiliary variable and $x$ is the only state variable of the duopoly, that is, the quantity dynamics of the duopoly is generated by the one-dimensional map $x^{\prime}=f\circ f\left(x\right)$. This map is the second iterated of $f$ and its graph is observable in Figures \ref{Fig::1-chaotic}(a), \ref{Fig::2-chaotic}(a) and \ref{Fig::4-chaotic}(a), for different configurations of the uncertainty set. Since $f\circ f$ exhibits chaotic dynamics when $f$ exhibits chaotic dynamics, we can conclude that chaotic dynamics persist in a duopoly game with ambiguity aversion where one firm acts sophisticatedly and the other one acts naively.

\section{Conclusions}\label{Sec::C}

In this paper, we have shown that combining ambiguity aversion, accommodated according to the worst-case approach to uncertainty, with a constant expectation scheme can compromise the stability of the Cournot-Nash equilibrium of a duopoly game and lead to chaotic dynamics. Specifically, assuming uncertainty about the degree of substitutability of the products and on the slopes of the inverse demand functions, a key result is the loss of stability of the Cournot-Nash equilibrium when firms are uncertain between a duopoly game and an oligopoly game with more than three firms.

This finding resembles Theocharis' result stating that the existence of more than three firms precludes the stability of the Cournot-Nash equilibrium in an oligopoly game without uncertainty, see, e.g., \cite{Theocharis1960}. In this respect, we conjecture that this standard result can be extended to games with uncertainty. More specifically, we conjecture that enriching an oligopoly-game setting as in \cite{Theocharis1960} with uncertainty does not favor stability when there are more than three competitors, but it can introduce chaotic dynamics. This conjecture will be the subject of future research.

The conducted investigation of the duopoly with ambiguity aversion is confined to the assumptions of a symmetric game and constant expectations. Relaxing the assumption of a symmetric setting, the dynamics of the model remain generated by a decoupled square system and a preliminary conducted analysis (not reported for the sake of space) confirms that coexisting Cournot-Nash equilibria and chaotic dynamics persist in a duopoly game with asymmetric uncertainty sets. Therefore, our results are robust. By contrast, introducing alternative expectation schemes, such as adaptive expectations, the dynamics of the duopoly game with ambiguity aversion is not described by a decoupled square system anymore and results and techniques for two-dimensional piecewise-smooth systems apply, see, e.g., \cite{SushkoGardini2008} and \cite{SushkoGardini2010}. The configuration of the player's uncertainty set is the third crucial assumption for the results obtained. Indeed, chaotic dynamics are not observed assuming an uncertainty set made by a single worst-case realization. However, relying on a single worst-case realization requires precise knowledge of the possible realizations of the values of the unknown parameters. Such knowledge is hardy observable in the real world and a conservative and consistent belief about the uncertainty set justifies two worst-case realizations. Finally, the generalization to an uncertainty set with more than two worst-case realizations is expected to increase the likelihood of chaotic dynamics.

All in all, the current work is the first attempt to analyze the effect of ambiguity aversion in stylized games in which the perfect foresight assumption about the action of the opponents is relaxed. Showing that ambiguity aversion can introduce very complicated dynamics even in very simple game-theory models is a remarkable finding that further underlines the effects of ambiguity aversion and paves the way for further research. Focusing on duopoly games, a possible further extension of the current work could be devoted to analyzing the persistence of chaotic dynamics when an adaptive expectation scheme is employed instead of a constant one. The effects of weaker forms of ambiguity aversion are also worth investigating.

\bigskip

\textbf{Acknowledgements:} We are grateful to Sushil Bikhchandani (editor), an anonymous Advisory Editor and anonymous Referees for helpful comments. The work was conducted within the research project on ``Models of behavioral economics for sustainable development'' founded by the Department of Economics, Society and Politics (DESP) of the University of Urbino. Davide Radi acknowledges the support of the Czech Science Foundation (GACR) under project 20-16701S and V\v{S}B-TU Ostrava under the SGS project SP2022/4. The authors acknowledge comments received from the participants to the Nonlinear Economic Dynamics (NED - 2021) conference in Milan.

\bigskip

\textbf{Declarations of interest:} None.

\bigskip

\textbf{Compliance with Ethical Standards:}
The authors declare no conflict of interest.


\appendix
\section{}\label{Appendix:Proof}

In this appendix we recall some properties regarding periodic cycles of map $T$ given in \eqref{DST}-\eqref{DST2}, and in general of a map having the following
structure $\left(x,y\right)\rightarrow\left(f\left(y\right),f\left(x\right)\right)$. These properties already considered in \cite{BischiMammanaGardini2000} and \cite{TramontanaGardiniPuu2010}, and taylored here for the scope of the current paper, are very useful because they show that we can obtain the coordinates of the periodic points of $T$ considering the periodic points of the one-dimensional map $f$.

Let us start by classifying cycles of $T$ in \emph{singly-generated}, so-called because their existence for map $T$ is a direct consequence of the existence of \emph{one} unique cycle of $f$, and \emph{doubly-generated}, so-called because associated with a pair of cycles of $f$. 

The following Property \ref{PropertyA3} classifies the singly-generated cycles of $T$.

\begin{property}[singly-generated cycles]\label{PropertyA3}
Let $\Delta$ be the diagonal $x=y$ and let $\left\{x_{1},x_{2},...,x_{n}\right\}$ be a cycle of $f$ of first period $n\geq 1$:
\begin{itemize}
\item If $n$ is odd then $T$ has: (a) one cycle of period $n$ on $\Delta$; (b) $\left(n-1\right)/2$ cycles of period $2n$ external to $\Delta$.
\item If $n$ is even and $n/2$ is also even then $T$ has $n$ cycles of period $n$, one on $\Delta$ and $\left(n-1\right)$ external to $\Delta$.
\item If $n$ is even and $n/2$ is odd then $T$ has: (a) $2$ cycles of period $n/2$ external to $\Delta$; (b) $\left(n-1\right)$ cycles of period $n$, one of which on $\Delta$.
\end{itemize}
\end{property}

\medskip

\begin{proof}[Proof of Property \ref{PropertyA3}]
Let $\left\{ x_{1},x_{2},...,x_{n}\right\} $ be a cycle of $f$ of first
period $n$ ($x_{i+1}=f\left(x_{i}\right)$ and $f^{n}\left(x_{i}\right)=x_{i}$ for $i=1,\ldots,n$) and consider the points of the Cartesian product $\left\{x_{1},\ldots,x_{n}\right\} \times \left\{x_{1},\ldots,x_{n}\right\}$. Moreover, consider the $k$-iterate of $T$, given by
\begin{equation}
T^{k}\left(x_{i},x_{j}\right)=\left\{ 
\begin{array}{lcr}
\left(f^{k}\left(x_{j}\right),f^{k}\left(x_{i}\right)\right) & \text{if} & k \ \text{is odd}  \\
\\
\left(f^{k}\left(x_{i}\right),f^{k}\left(x_{j}\right)\right) & \text{if} & k \ \text{is even}
\end{array}
\right.   \label{TPR}
\end{equation}
If $i=j$ we have a point on $\Delta$ and thus the first integer giving a cycle is $k=n$ (and we get the $n-$cycle on $\Delta$), while for $i\neq j$ we have a point $\left(x_{i},x_{j}\right)$ external to $\Delta$, and the first integer giving a cycle depends on the period $n$. If $n$ is odd, then the first integer giving a periodic point in \eqref{TPR} is $k=2n$ so that $\left(x_{i},x_{j}\right)$ belongs to a cycle of $T$ external to $\Delta$ of first period $2n$. Such distinct cycles must equal $\left(n^{2}-n\right)/\left(2n\right)=\left(n-1\right)/2$ in number. If $n$ is even, then $T^{k}\left(x_{i},x_{j}\right) = \left(f^{k}\left(x_{j}\right),f^{k}\left(x_{i}\right)\right) = \left(x_{i},x_{j}\right)$ may occur for $k$ odd and such that $T^{2k}\left(x_{i},x_{j}\right) = \left(f^{2k}\left(x_{i}\right),f^{2k}\left(x_{j}\right)\right)=\left(x_{i},x_{j}\right)$. This indicates that either a prime period $k=n/2$ exists, with $k$ odd and the periodic points belonging to two distinct cycles of period $n/2$ are $\left(x_{i},x_{i+n/2}\right)$ and $\left(x_{i+n/2},x_{i}\right)$, or 
$T^{k}\left(x_{i},x_{j}\right)  = \left( f^{k}\left(x_{i}\right),f^{k}\left(x_{j}\right)\right) = \left(x_{i},x_{j}\right)$ occurs for $k$ even which leads to a prime period $k=n$ external to $\Delta$, and at most $\left(n^{2}-n\right)/n=\left(n-1\right)$ distinct cycles of period $n$ can exist.
\end{proof}

\medskip

The following Property \ref{PropertyA4} classifies the doubly-generated cycles of $T$ associated with each pair of cycles of $f$. Let us recap that a cycle of $T$ is denoted doubly-generated because each point of the cycle has the coordinates belonging to two different cycles of $f$.

\begin{property}[doubly-generated cycles]\label{PropertyA4}
Let $\left\{x_{1},...,x_{n}\right\}$ be a cycle of $f$ of first period $n\geq 1$, and $\left\{ y_{1},...,y_{m}\right\}$ a cycle of $f$ of first period $m\geq 1$, and let $L$ be the least common multiple between $n$ and $m$, then the cycles of $T$ of type doubly-generated are as follows:
\begin{itemize}
\item If $n$ and $m$ are odd then $T$ has $\frac{n\cdot m}{L}$ cycles of period $2L$;
\item If $n$ or/and $m$ are even then $T$ has $\frac{2n\cdot m}{L}$ cycles of period $L$.
\end{itemize}
\end{property}

\medskip

\begin{proof}[Proof of Property \ref{PropertyA4}]
Consider the $k$-iterate of $T$, given by
\begin{equation}\label{TPP}
T^{k}\left(x_{i},y_{j}\right)=\left\{ 
\begin{array}{lcr}
\left(f^{k}\left(y_{j}\right),f^{k}\left(x_{i}\right)\right) &  \text{if} & $k$  \ \text{ is odd} \\
\\ 
\left(f^{k}\left(x_{i}\right),f^{k}\left(y_{j}\right)\right) & \text{if} &  $k$  \ \text{ is even}
\end{array}
\right.
\end{equation}
It is clear that when $L$ is odd (which can occur only when both $n$ and $m$ are odd), then the least integer giving a periodic point in \eqref{TPP} is $k=2L$, and we get a cycle of $T$ of period $2L.$ Such cycles may be $2\left(n\cdot m\right)/\left(2L\right)=n\cdot m/L$ in number$.$ When $L$ is even (which occurs when $n$ or/and $m$ are even), then the least integer giving a periodic point in \eqref{TPP} is $k=L$, so that we get a cycle of $T$ of period $L$, and such cycles may be in number $2\left(n\cdot m\right)/L$.
\end{proof}

\medskip

Regarding the stability/instability of the cycles of $T$, we make use of the fact that any cycle of $T$ is related to one (if singly-generated) or two (if doubly-generated) cycles of $f$, and it is easy to see, considering the Jacobian matrix of $T$, that the following Property \ref{PropertyA5} holds. The proof of the following Property \ref{PropertyA5} is omitted as it is straightforward.

\begin{property}\label{PropertyA5}
Let $X=\left\{x_{1},...,x_{n}\right\}$ be a cycle of $f$ of first period $n\geq 1$, and $Y=\left\{y_{1},...,y_{m}\right\}$ a cycle of $f$ of first period $m\geq 1$, then:
\begin{itemize}
\item If $X$ is locally asymptotically stable (resp. unstable) for $f$ with eigenvalue $\lambda$, such that $\left|\lambda\right|<1$ (resp. $\left|\lambda \right|>1$), then all the singly-generated cycles associated with $X$ are locally asymptotically stable (resp. unstable) nodes for $T$ with eigenvalues $\zeta _{1}=\zeta _{2}=\lambda$.
\item If $X$ and $Y$ are both locally asymptotically stable for $f$, with eigenvalues $\lambda$ and $\mu$, such that $\left|\lambda\right|<1$ and $\left|\mu\right|<1$, then all the doubly-generated cycles associated with $X$ and $Y$ are locally asymptotically stable nodes for $T$ with eigenvalues $\zeta_{1}=\lambda$ and $\zeta _{2}=\mu$.
\item If either $X$ or $Y$ is unstable for $f$, with eigenvalues $\lambda$ and $\mu$, such that $\left|\lambda\right|<1$ and $\left|\mu\right|>1$, then all the doubly-generated cycles associated with $X$ and $Y$ are unstable for $T$, of saddle type, with eigenvalues $\zeta_{1}=\lambda$ and $\zeta_{2}=\mu$.
\item If $X$ and $Y$ are both unstable for $f$, with eigenvalues $\lambda$ and $\mu$, such that $\left|\lambda\right|>1$ and $\left|\mu\right|>1$, then all the doubly-generated cycles associated with $X$ and $Y$ are unstable nodes for $T$, with eigenvalues $\zeta _{1}=\lambda$ and $\zeta_{2}=\mu$.
\end{itemize}
\end{property}



\end{document}